%% file: center-curve.tex
\documentclass[a4paper, UKenglish]{lipics-v2018}
\usepackage{microtype,xspace}
\usepackage{multirow}
\usepackage{amsfonts,amsmath,amsthm,nicefrac}
\usepackage{thm-restate}
\usepackage{graphicx}
\graphicspath{{./figures/}}
\linenumbers
\usepackage{color}
\usepackage{paralist}

\nolinenumbers

\title{Approximating $(k,\ell)$-center clustering for curves}
\titlerunning{Approximating $(k,\ell)$-center clustering for curves} 

\author{Kevin Buchin}{Department of Mathematics and Computing Science, TU Eindhoven, The Netherlands}{k.a.buchin@tue.nl}{}{}
\author{Anne Driemel}{Department of Mathematics and Computing Science, TU Eindhoven, The Netherlands}{a.driemel@tue.nl}{}{}
\author{Joachim Gudmundsson}{School of Information Technologies, The University of Sydney, Australia}{joachim.gudmundsson@sydney.edu.au}{}{}
\author{Michael Horton}{ 
    Tandon School of Engineering, New York University, USA\\School of Information Technologies, The University of Sydney, Australia}{michael.horton@nyu.edu}{}{} 
\author{Irina Kostitsyna}{Department of Mathematics and Computing Science, TU Eindhoven, The Netherlands}{i.kostitsyna@tue.nl}{}{}
\author{Maarten L\"offler}{Department of Information and Computing Sciences, Utrecht University, The Netherlands}{m.loffler@uu.nl}{}{}
\author{Martijn Struijs}{Department of Mathematics and Computing Science, TU Eindhoven, The Netherlands}{m.a.c.struijs@student.tue.nl}{}{}

\authorrunning{Buchin, Driemel, Gudmundsson, Horton, Kostitsyna, L\"offler, and Struijs} 

\relatedversion{}
\funding{}
\acknowledgements{}

\Copyright{Kevin Buchin, Anne Driemel, Joachim Gudmundsson, Michael Horton, Irina Kostitsyna, Maarten L\"offler and Martijn Struijs}

\subjclass{Theory of computation/Randomness, geometry and discrete structures/Computational geometry}

\hideLIPIcs  

\newcommand{\R}{\mathbb{R}\xspace}
\providecommand{\eps}{\varepsilon}

\newcommand{\fd}{Fr\'echet distance\xspace}
\newcommand{\ls}{$\ell$-simplification\xspace}

\newcommand{\klc}[1][k]{$(#1,\ell)$-\textsc{Center}\xspace}

\newcommand{\myparNS}[1]{\noindent{\sffamily\bfseries #1.}}
\newcommand{\mypar}[1]{\medskip\myparNS{#1}}

\begin{document}

\maketitle

\begin{abstract}
\input{abstract}
\end{abstract}


\input{introduction}

\input{results}
\input{preliminaries}
\input{summary}

\input{hardness1D}

\input{hardness2D}

\input{hardness2D-appendix}

\input{hardness-MEB}
\input{approx}

\input{approx-appendix}

\input{acknowledgements}


\bibliographystyle{abbrv}
\bibliography{refs}

\clearpage
\appendix
\input{hardness1D-appendix}


%
\end{document}

%% file: abstract.tex
The Euclidean $k$-center problem is a classical problem that has been extensively studied in computer science. Given a set $\mathcal{G}$ of $n$ points in Euclidean space, the problem is to determine a set $\mathcal{C}$ of $k$ centers (not necessarily part of $\mathcal{G}$) such that the maximum distance between a point in $\mathcal{G}$ and its nearest neighbor in $\mathcal{C}$ is minimized. In this paper we study the corresponding $(k,\ell)$-center problem for polygonal curves under the Fr\'echet distance, that is, given a set $\mathcal{G}$ of $n$ polygonal curves in $\R^d$, each of complexity $m$, determine a set $\mathcal{C}$ of $k$ polygonal curves in $\R^d$, each of complexity $\ell$, such that the maximum Fr\'echet distance of a curve in $\mathcal{G}$ to its closest curve in $\mathcal{C}$ is minimized.
In their 2016 paper, Driemel, Krivo{\v{s}}ija, and Sohler give a near-linear time $(1+\eps)$-approximation algorithm for one-dimensional curves, assuming that $k$ and $\ell$ are constants. In this paper, we substantially extend and improve the known approximation bounds for curves in dimension $2$ and higher. Our analysis thus extends to application-relevant input data such as GPS-trajectories and protein backbones.
We show that, if $\ell$ is part of the input, then there is no polynomial-time approximation scheme unless $\mathsf{P}=\mathsf{NP}$. Our constructions yield different bounds for one and two-dimensional curves and the discrete and continuous Fr\'echet distance. In the case of the discrete Fr\'echet distance on two-dimensional curves, we show hardness of approximation within a factor close to $2.598$.
This result also holds when $k=1$, and the $\mathsf{NP}$-hardness extends to the case that $\ell=\infty$, i.e., for the problem of computing the minimum-enclosing ball under the Fr\'echet distance. 
Finally, we observe that a careful adaptation of Gonzalez' algorithm in combination with a curve simplification yields a $3$-approximation in any dimension, provided that an optimal simplification can be computed exactly. We conclude that our approximation bounds are close to being tight.

%% file: introduction.tex
\section{Introduction}
Clustering is a fundamental task in data analysis. 
Grouping similar objects together enables the efficient summarization of large amounts of data, and the discovery of hidden patterns. A classical clustering problem is the Euclidean $k$-{\sc Center} problem~\cite{ap-eaac-02,bhi-accs-02,hn-ehblp-79}. Given a set $\mathcal{G}$ of $n$ points in Euclidean space, the problem is to determine a set $\mathcal{C}$ of $k$ centers (not necessarily part of $\mathcal{G}$) such that the maximum distance between a point in $\mathcal{G}$ and its nearest neighbor in $\mathcal{C}$ is minimized.
In general, most research on clustering focuses on point sets, and while many data can be interpreted as points, point clustering is less appropriate for sequentially-recorded data like time series and trajectories. Clustering such types of data is an active research topic~\cite{jp-fdc-14}, and in many cases generalizes the approach of assigning the data to the nearest cluster center from points to curves~\cite{acmm-uccub-03,cmp-fcis-07,gg-prcc-05,pg-ssts-12,pkg-gadtw-11}. However, despite considerable research on the problem of curve clustering, little is known about the algorithmic complexity of this problem.

In this paper we study approximation algorithms for the generalization of $k$-{\sc Center} clustering to curves.
As a distance measure between curves, the discrete and continuous Fr\'echet distance is used. The Fr\'echet distance~\cite{ag-cfdbt-95} is an effective distance measure for polygonal curves, as it takes into account the continuity of the curves, and it also works well for irregularly sampled curves. The computational complexity of the Fr\'echet distance has received considerable attention recently~\cite{bk-iafd-17,bm-adfd-16,bbmm-fswd-17}.
We will denote the continuous Fr\'echet distance between two curves $\gamma$ and $\gamma'$ in $\R^d$ by $d_F(\gamma,\gamma')$ and the discrete Fr\'echet distance by $d_{DF}(\gamma,\gamma')$.
In cases where either measure is valid, we use $d(\gamma,\gamma')$.

We consider the following problem:
Given a set of curves $\mathcal{G}$ in $\R^d$, the \klc problem asks to find the minimum (discrete or continuous) Fr\'echet distance $\delta$ for which there exist $k$ center curves, each of complexity at most $\ell$, 
such that each input curve is within 
distance $\delta$ of at least one of the center curves. In other words, there exist $k$ curves $\{c_1,\dots,c_k\}$ in $\R^d$ each of complexity no greater than $\ell$, such that for any curve $\gamma\in\mathcal{G}$ there exists an $1\leq i\leq k$ with $d(\gamma,c_i)\leq\delta$. 

Restricting the complexity of the center curves 
improves the relevance of the obtained 
clustering for applications (see Figure~\ref{fig:pigeon-example}), and the performance of algorithms that make use of the resulting centers~\cite{pg-ssts-12,pkg-gadtw-11}. Indeed, by not restricting the complexity, overfitting will occur since the center curve is susceptible to picking up all the details of the input curves~\cite{aabosw-mcdfd-16,dks-cts-16}.

\begin{figure}[b]
  \begin{minipage}[c]{0.51\textwidth}
    \includegraphics[width=\textwidth]{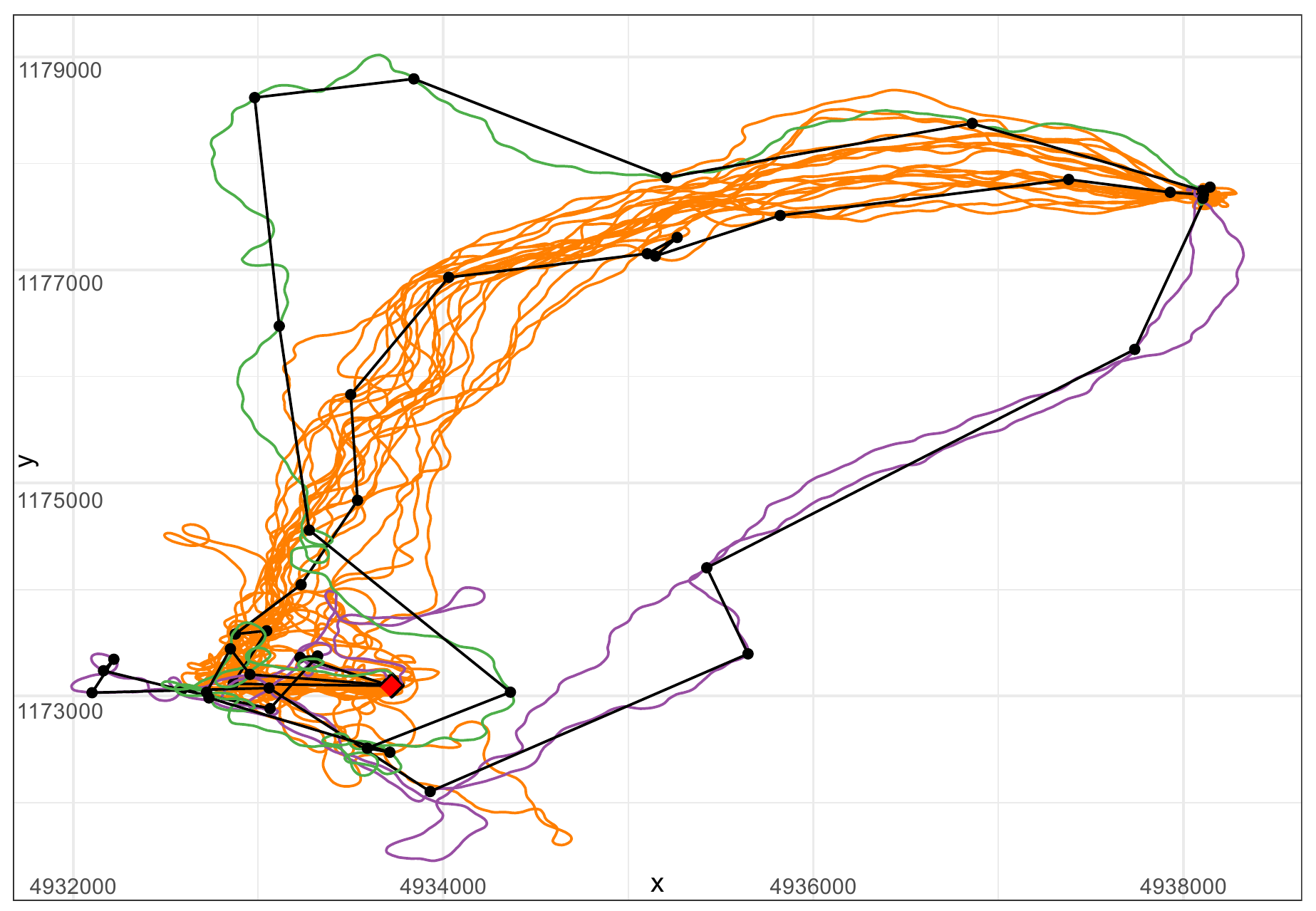}
  \end{minipage}\hfill
  \begin{minipage}[c]{0.47\textwidth}
\caption{$(3,15)$-center clustering of pigeon flight based on the algorithm in this paper. Data from~\cite{Mann20130885}, clustering by Natasja van de l'Isle. It has been posited that pigeons fly home following a series of landmarks. Therefore, low-complexity cluster centers 
      not only 
      summarize the data without overfitting, but may also aid in identifying important landmarks~\cite{Mann210}. The clusters are shown in green, orange and purple, the cluster centers are drawn in black. The pigeons flew from a release site (red diamond) to their home loft.}
\label{fig:pigeon-example}
  \end{minipage}
\end{figure}

The \klc problem was introduced and studied algorithmically by Driemel, Krivo\v{s}ija and Sohler~\cite{dks-cts-16} for one-dimensional curves. 
They show that the problem is $\mathsf{NP}$-hard and provide a $(1+\eps)$-approximation algorithm with a running time near-linear in the input size assuming $k,\ell \in O(1)$. They 
also describe a near-linear time constant-factor approximation algorithm that works in any dimension.

%% file: results.tex
\mypar{Our results}
 In this paper we significantly improve and extend the results of Driemel et~al.\ for the \klc.
%
    We show that the 
    \klc{} problem \emph{where $\ell$ is part of the input} is $\mathsf{NP}$-hard, even if $k=1$ and $d = 1$. 
    In particular, we prove hardness of approximation within different constant factors for various cases, improving the results by Driemel et~al. The bounds are summarized in Table~\ref{table:hardness:results}. Moreover, in Section~\ref{sec:meb} we show that when relaxing the complexity of the center curve, i.e., for $\ell=\infty$, the problem remains $\mathsf{NP}$-hard. Thus, the problem of computing the \emph{minimum enclosing ball} under the Fr\'echet distance is also hard. Note that any of the input curves serves as a $2$-approximation to the unconstrained $1$-{\sc center}.
    Interestingly, for the constrained case, we can show approximation bounds well above the factor $2$ if the curves lie in the plane. This shows that our problem is strictly harder than the $k$-{\sc center} problem in general discrete metric spaces. 

\begin{table}[t]
\centering
    \caption{The table shows new results on approximation factors for which it is hard to approximate the \klc problem unless $\mathsf{P}=\mathsf{NP}$, the bounds hold for any $\eps>0$ and even if $k=1$. }
\label{table:hardness:results}
\renewcommand{\arraystretch}{1.1}
\begin{tabular}{|l|ll|ll|}
\hline
    & \multicolumn{2}{c|}{Discrete \fd} & \multicolumn{2}{c|}{Continuous \fd} \\
 \hline
    $d = 1$ & 
    $2-\eps$ & (Theorem~\ref{thm:1dhardness:discrete}) &  
    $1.5-\eps$ & (Theorem~\ref{thm:1dhardness:cont}) \\
\hline
    $d \geq 2$ & 
    $2.598$ & (Theorem~\ref{thm:2dhardness:discrete}) & 
    $2.25-\eps$ & (Theorem~\ref{thm:2dhardness:continuous}) \\
\hline
\end{tabular}
\end{table}

\begin{table}[t]
\centering
    \caption{Summary of the running times of approximation algorithms for \klc problem. The notation $\widetilde{O}(\cdot)$ hides logarithmic factors.}
    \renewcommand{\arraystretch}{1.1}
\label{table:algorithms:results}
    \begin{tabular}{|l|c|cc|c|cc|}
    \hline
    & \multicolumn{3}{c|}{Discrete \fd} & \multicolumn{3}{c|}{Continuous \fd} \\
        \cline{2-7}
        & 
        \multicolumn{1}{c|}{Time} & \multicolumn{2}{c|}{Approximation} &
        \multicolumn{1}{c|}{Time} & \multicolumn{2}{c|}{Approximation} \\
 \hline
    $d=1$ & 
        \multicolumn{1}{c|}{--} & \multicolumn{2}{c|}{--} & 
        $\widetilde{O}(mn)$\footnotemark[2] & $1+\eps$ & \cite{dks-cts-16} \\
    \hline
    $d\geq 1$ & 
        ${O}(poly(m,n,k,\ell))$\footnotemark[3] & $5$ & \cite{dks-cts-16} &
        ${O}(poly(m,n,k,\ell))$ & $8$ & \cite{dks-cts-16} \\
    \hline
    $d=2$ & 
        \multicolumn{1}{c|}{--} & \multicolumn{2}{c|}{--} & 
        $\widetilde{O}(km(n \ell + m^2))$ & $3$ & (Theorem~\ref{cor:alg:impr}) \\
    \hline
    $d\geq 2$ & 
        $\widetilde{O}(kn \cdot \ell m )$ & $3$ & (Corollary~\ref{cor:alg:basic}) & 
        $\widetilde{O}(km(n \ell + m^2))$ & $6$ & (Corollary~\ref{cor:alg:basic}) \\
    \hline
\end{tabular}
\end{table}
\footnotetext[2]{{Assuming $k,\ell \in O(1)$.}}
\footnotetext[3]{{Running time is not explicitly stated in the paper.}}

    At the same time, our algorithmic results show that one can efficiently compute a $3$-approximation in this case.
    The $(1+\eps)$-approximation algorithm described by  Driemel~et~al.~\cite{dks-cts-16} has running time exponential in both parameters~$k$ and~$\ell$. Our hardness results indicate that the exponential dependency on $\ell$ may be unavoidable for small approximation factors. Driemel~et~al.\ also describe a faster constant-factor approximation algorithm that uses Gonzalez' algorithm in combination with a simplification algorithm, achieving an approximation factor of $\alpha+\beta+\alpha\beta$, where $\alpha$ is the approximation factor of the simplification and $\beta$ is the approximation factor of the metric clustering algorithm. Our algorithm uses a slightly different adaptation of Gonzalez' algorithm, and we analyze the scheme in spaces of arbitrary dimension, and for both the discrete and continuous \fd{}s. 
The approximation factors and running times of these algorithms are summarized in Table~\ref{table:algorithms:results}.




The paper is organized so as to provide an quick overview of the results and techniques with additional details in later sections. Section~\ref{sec:basic:algorithm:sketch} provides a sketch of our algorithm with additional details in Sections~\ref{sec:clustering} and~\ref{sec:imp-apx-cl}. A sketch of the 1D hardness construction is given in Section~\ref{sec:1d}, the full 2D construction is given in Section~\ref{sec:2dhard} and the extension to the minimum enclosing ball problem in Section~\ref{sec:meb}.

%% file: preliminaries.tex
\section{Preliminaries}

Let $\gamma$ be a polygonal curve given by a sequence of vertices $P := \langle p_1,\dotsc,p_m \rangle$. The curve is defined as a parametric curve through the vertices, connecting each contiguous pair of vertices in $P$ by a straight edge $\overline{p_i p_{i+1}}$.
Let the \emph{complexity} $\vert \gamma \vert$ of $\gamma$ be the number of vertices in the sequence, i.e. $\vert P \vert = m$.

The \emph{\fd} measures the distance between two curves $\gamma$ and $\gamma'$, and we consider two variants.
The \emph{continuous} \fd is defined using a reparameterization $f \colon [0,1] \to [0,1]$ that is a  continuous injective function where $f(0) \equiv 0$ and $f(1) \equiv 1$. Let $\mathcal{F}$ be the family of all such reparameterizations, then the continuous \fd is defined as \[d_F(\gamma,\gamma') = \inf_{f \in \mathcal{F}} \max_{\alpha \in [0,1]} \lVert \gamma(f(\alpha)) - \gamma'(\alpha) \rVert,\]
where $\lVert \cdot \rVert$ is the Euclidean norm.

The \emph{discrete} \fd is defined by an alignment between the vertex sequences $P$ and $P'$ of two curves.
Let $T$ be a sequence of pairs of indices $\langle (i_1,j_1),\cdots,(i_t,j_t)) \rangle$ where $i_1 = 1, j_1 = 1, i_t = \vert P \vert, \text{ and }j_t = \vert P' \vert$.
Let $\mathcal{T}$ be the family of all such alignments $T$.
For each pair $(i_s,j_s)$ where $1 < s \leq t$, one of the following holds:
\begin{inparaenum}[(i)]
    \item{$i_s = i_{s-1}$ and $j_s = j_{s-1} + 1$, or}
    \item{$i_s = i_{s-1} + 1$ and $j_s = j_{s-1}$, or}
    \item{$i_s = i_{s-1} + 1$ and $j_s = j_{s-1} + 1$.}
\end{inparaenum}
The discrete \fd is defined as:
\[ d_{DF}(\gamma, \gamma') = \inf_{T \in \mathcal{T}} \max_{(i,j) \in T} \lVert p_i - p_j'\rVert .\]


The discrete \fd is metric, and the continuous \fd is \emph{pseudo-metric} as two distinct curves may have zero distance~\cite{ag-cfdbt-95}.
However, for our purposes it suffices that the triangle inequality and symmetry properties are satisfied.

The key to our approximation results is a careful adaptation of Gonzalez' algorithm in combination with curve simplification algorithms.
Given a curve $\gamma$, a \emph{simplification} $\overline{\gamma}$ is a curve defined by a vertex sequence $\langle v_1,\dotsc,v_{\ell} \rangle$ such that the complexity of $\overline{\gamma}$ is less than that of $\gamma$ and the distance $d$---or \emph{error}---between the $\gamma$ and $\overline{\gamma}$ is small.
Simplification is a bicriteria operation, and thus there are two optimisation problems: 
\begin{compactitem}[(i)]
\item{\emph{minimum-complexity $\eps$-simplification:} given $\gamma$ and an error $\eps > 0$, find the simplification $\overline{\gamma}$ such that $d(\gamma,\overline{\gamma}) \leq \eps$ and $\vert \overline{\gamma} \vert$ is minimised, and}
\item{\emph{minimum-error $\ell$-simplification:} given $\gamma$ and an integer $0 < \ell < m$, find the simplification $\overline{\gamma}$ such that $\vert \overline{\gamma} \vert \leq \ell$ and the error $d(\gamma,\overline{\gamma})$ is minimised.}
\end{compactitem}
The problems are also referred to as min-$\#$ and min-$\eps$ respectively in some works~\cite{bjwyz-spcdfd-08,g-anmfc-91,ii-pac-88}.

There are many variants of the simplification problem.
In particular, the simplification may be 
\emph{vertex-constrained}, i.e., the vertex sequence of $\overline{\gamma}$ is required to be a subsequence of the vertex sequence of $\gamma$ with $v_1 = p_1$, $v_\ell = p_m$ and $d(\gamma, \overline{\gamma}) = \max_{1\leq i < \ell} 
\overline{d} (\overline{v_i,v_{i+1}},\gamma[v_i,v_{i+1}])$, where $\gamma[v_i,v_{i+1}]$ is the subcurve of $\gamma$ between $v_i$ and $v_{i+1}$, and $\overline{d}$ is $d_F$, $d_{DF}$ (or in previous work~\cite{dp-arnp-73,ii-pac-88} also
the Hausdorff distance).
In this paper, if not stated otherwise, we
consider \emph{weak}
simplifications~\cite{ahmw-nltaa-05}, i.e., simplifications with no restrictions (and with $d$ as $d_F$ or $d_{DF}$).

Curve simplification is a well-studied problem that stands on its own in computational geometry.
Two early contributions were the simplification algorithms by Douglas and Peuker~\cite{dp-arnp-73} and Imai and Iri~\cite{ii-pac-88}, which both induce a vertex-constrained simplification. 
Godau~\cite{g-anmfc-91} studied the problem of computing a simplification under the \fd, and observed that a simplification is a bicriteria approximation. 
Godau presented modifications to the simplification algorithm in~\cite{ii-pac-88} to solve the vertex-constrained min-$\#$ and min-$\eps$ problems in $O(m^3)$ and $O(m^4\log{m})$ time, respectively. 
Furthermore, the result from the min-$\eps$ algorithm was shown to be a $7$-approximation of the optimal \emph{weak} simplification. 
Agarwal~et~al.~\cite{ahmw-nltaa-05} subsequently improved this approximation bound to $4$.

Guibas~et~al.~\cite{ghms-91} presented an $O(n^2 \log^2{n})$ time algorithm for computing the min-$\#$ weak simplification under the continuous \fd in $\R^2$, and Bereg~et~al.~\cite{bjwyz-spcdfd-08} gave an $O(n\log{n})$ algorithm for the min-$\#$ problem, and a $O(\ell m \log{m} \log{(m/\ell)})$ algorithm for the min-$\eps$ problem, in the discrete \fd and in spaces of arbitrary dimension.

%% file: summary.tex
\section{Sketch of the Basic Algorithm}\label{sec:basic:algorithm:sketch}

Our basic algorithm for the \klc problem is a careful adaptation of  the clustering algorithm for points in a metric space by Gonzalez~\cite{g-cmmid-85}. 
The general idea of Gonzalez' algorithm is to iteratively compute a set $\mathcal{C}$ of $k$ centers, where $\mathcal{C}$ 
is a subset of the input points.  
In the first iteration, the algorithm selects an arbitrary input point as a center by adding it to the (initially empty) set of centers $\mathcal{C}$.
In each of the subsequent $k-1$ iterations, the input point that is farthest from all the centers selected so far is identified and added to $\mathcal{C}$. 
Gonzalez~\cite{g-cmmid-85} 
showed that this simple greedy approach is a $2$-approximation algorithm for the $k$-center clustering problem on points in a metric space.

We modify Gonzalez' algorithm to compute an approximate solution to the \klc problem for curves under the \fd.
The algorithm operates in essentially the same way as the original Gonzalez' algorithm but on a set of curves as input. 
The distance between the curves is measured using the \fd and, in each iteration, the $\ell$-simplification of the selected curve is added to $\mathcal{C}$.
The details of this algorithm are described in Section~\ref{sec:clustering}.
The approximation factor of the algorithm can be improved using standard techniques, refer to Section~\ref{sec:imp-apx-cl} for the details. We get the following algorithmic result.

\begin{restatable}{theorem}{improvedalgresult}
\label{cor:alg:impr}
    Given $n$ polygonal input curves in the plane, each of complexity $m$, and positive integers $k,\ell$, one can compute a $3$-approximation to the \klc problem in time $O(km(n \ell \log(\ell+m)+ m^2\log{m}))$.
\end{restatable}

%% file: hardness1D.tex
\section{Sketch of the Basic Hardness Reduction}\label{sec:1d}

The \klc problem is $\mathsf{NP}$-hard for polygonal curves that lie in one dimension. This holds if distances are measured under the discrete or the continuous Fr\'echet distance. Furthermore, the \klc problem is hard to approximate within a factor $2-\eps$ for the discrete Fr\'echet distance, and within a factor $1.5-\eps$ for the continuous Fr\'echet distance, for any value of $\eps>0$. In this section we will describe the main idea of the reduction for the discrete Fr\'echet distance. Refer to Appendices~\ref{sec:1d-app} and~\ref{sec:1Dappx-app} for the full details and for the reduction in the case of continuous Fr\'echet distance.

We reduce from the {\sc Shortest Common Supersequence} ({\sc SCS}) problem, which asks to compute a shortest sequence $s^*$ such that each of $n$ finite input strings $s_i$ over a finite alphabet $\Sigma$ is a subsequence of $s^*$. This problem is known to be NP-hard for binary alphabets~\cite{RU81}. Given an instance of the decision version of the {\sc SCS} problem, i.e., a set of strings $S=\{s_1,s_2,\dots,s_n\}$ over an alphabet $\{A,B\}$, and a maximum allowed length $t$ of the sought supersequence, we will construct a corresponding instance of the decision version of \klc problem for Fr\'echet distance $\delta=1$ and for $k=1$.

For each input string $s_i\in S$ we construct a one-dimensional curve $\gamma(s_i)$ in the following way. The curve $\gamma(s_i)$ will have a vertex at $(-3)$ for each letter $A$ in $s_i$, and a vertex at $(3)$ for each letter $B$ in $s_i$. We call these vertices \emph{letter gadgets}. The letter gadgets in $\gamma(s_i)$ are separated by \emph{buffer gadgets}: sequences of vertices alternating between $(-1)$ and $(1)$; and $\gamma(s_i)$ also starts and ends with a buffer gadget. We choose a length of the buffer gadget at least twice the size of a sought superstring. The buffer gadget enables the Fr\'echet-matching to `skip' over any encoded letter in the center curve that is not present in $s_i$. Figure~\ref{fig:1dgadgetsD-example} shows an example of the reduction for input $ABB$, $BBA$, and $ABA$.

\begin{figure}[t]
\centering
    \includegraphics[page=2]{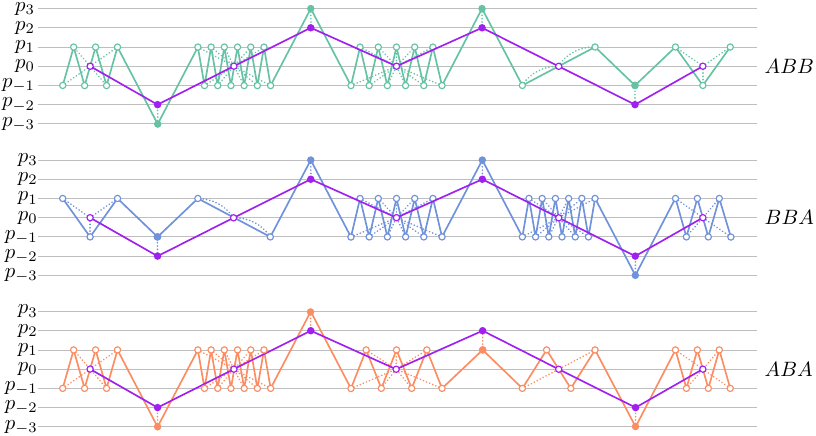}
\caption{Three curves corresponding to strings $ABB$ (green), $BBA$ (blue), and $ABA$ (orange), and a center curve (purple) within discrete Fr\'echet distance $1$ from them.}
\label{fig:1dgadgetsD-example}
\end{figure}

Having constructed a curve $\gamma(s_i)$ for each string $s_i\in S$, we show that there exists a string $s^*$ of length $t$ that is a supersequence of all $s_i$ if and only if there exists a \emph{center curve} $c^*$ with $\ell=2t+1$ vertices that lies within discrete Fr\'echet distance $\delta=1$ from all $\gamma(s_i)$.

First, let $s^*$ be such a supersequence. 
Observe that a curve $c^*$ 
that has a vertex at $(-2)$ for every letter $A$ in $s^*$, a vertex at $(2)$ for every letter $B$ in $s^*$, and with vertices at $(0)$ at the start and end of the curve, and also between each pair of letters, is within discrete Fr\'echet distance $\delta=1$ from all the curves $\gamma(s_i)$ (refer to Figure~\ref{fig:1dgadgetsD-example}).

Now let there exist a center curve $c$ of size $2t+1$ that is within discrete Fr\'echet distance $\delta=1$ from all the curves $\gamma(s_i)$. For every curve $\gamma(s_i)$, the curve $c$ must have a vertex (with absolute value $\ge 2$) per each letter gadget. To cover the buffer gadgets, $c$ must have at least one extra vertex per gadget, as the vertices covering letter gadget cannot completely cover the buffer gadgets as well. Thus, if $c$ has $2t+1$ vertices, then at most $t$ of these vertices can match to letter gadgets of curves $\gamma(s_i)$. A string $s^*$ consisting of letters corresponding to the vertices of $c$ that match to letter gadgets is a supersequence of all strings in $S$. This implies that the \klc problem in 1D is $\mathsf{NP}$-hard for the discrete Fr\'echet distance.
We show hardness of approximation using the same construction. We prove that, for any instance of the {\sc SCS} problem, if there exists a center curve of size $\ell$ within discrete Fr\'echet distance $\delta$ from all the curves $\gamma(s_{i})$, where $\delta < 2$, then there exists a center curve of size at most $\ell$ within discrete Fr\'echet distance $1$ from all the curves $\gamma(s_{i})$.
We state the resulting theorems here. Refer to the Appendix~\ref{sec:1Dappx-app} for the full details and proofs. 
\begin{restatable}{theorem}{hardoneddiscreteappx}
\label{thm:1dhardness:discrete}
    The \klc problem for polygonal curves in 1D for the discrete Fr\'echet distance is $\mathsf{NP}$-hard to approximate within approximation factor $2-\varepsilon$, even if $k=1$.
\end{restatable}

\begin{restatable}{theorem}{hardonedcontappx}
\label{thm:1dhardness:cont}
    The \klc problem for polygonal curves in 1D for the continuous Fr\'echet distance is $\mathsf{NP}$-hard to approximate within approximation factor $1.5-\varepsilon$, even if $k=1$.
\end{restatable}

This construction can be extended to the case of $\ell=\infty$, i.e. for the problem of computing the minimum enclosing ball under the Fr\'echet distance. The main idea is to add additional curves that will bound the length of the supersequence constructed from any valid center curve. We refer to Section~\ref{sec:meb} for the full construction and further details.



%% file: hardness2D.tex
\section{Hardness of Approximation in 2D}
\label{sec:2dhard}

We again describe a reduction from the {\sc SCS} decision problem to the
\klc problem. The construction described in this section shows 
improved results for hardness of approximation of the \klc 
problem for curves in the plane and in higher dimensions.
We give full details of the analysis for the discrete
Fr\'echet distance. The reduction in the continuous case uses the same gadgets. Refer to Appendix~\ref{sec:2Dhard-app} for the full details of the proof. 

\subsection{The Reduction}

\begin{figure}\centering
\includegraphics[page=1,width=0.3\textwidth]{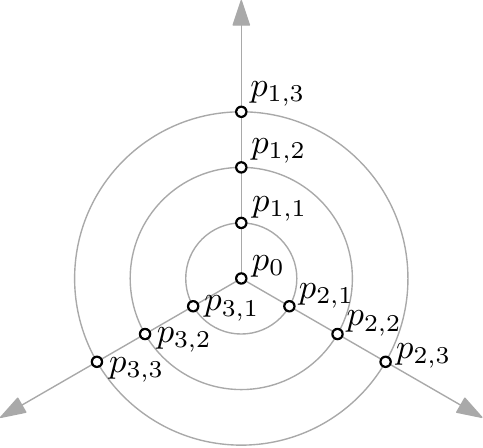}
\includegraphics[page=2,width=0.3\textwidth]{figures/2Dgadgets}
\includegraphics[page=3,width=0.3\textwidth]{figures/2Dgadgets}
\caption{The gadgets for the construction in the plane. Left: Notation
of vertices used in the construction. Center: Gadgets $g_A$ (black), buffer
gadget $g_a$ (blue). Right:  Gadgets $g_B$ (black), buffer gadget $g_b$ (blue).
The trace of the center curve used in the proof is shown in fat red in both
figures.} 
\label{fig:2dgadgets}
\end{figure}

An instance of the decision version of the {\sc SCS} problem is given as a set of
strings $S=\{s_1,s_2,\dots,s_n\}$, each of length at most $n$, over an alphabet
$\{A,B\}$, and a value of $t$---the maximum allowed length of the sought
superstring. For any such instance we construct a corresponding instance of the
\klc problem.  For each input string $s_i \in S$
we construct a curve $\gamma(s_i)$ in the plane and this construction is
described in this section. 
Our construction uses a global parameter $s \geq 1$, the value of which will be
fixed depending on the value of $t$, the parameter for the length of the superstring.  

First, we define a set of ten points in the plane that will be used in the
construction of the gadgets.  Let $p_0=(0,0)$ be the origin and
consider three circles of radius $1,2$ and $3$, which are each centered at
$p_0$.  Consider three rays from $p_{0}$ at $90,150$ and $210$ degrees. We define 
the remaining points $p_{i,j}$ as the intersections of these rays with the
three circles, where $i \in \{1,2,3\}$ indicates the ray and $j \in \{1,2,3\}$ 
indicates the circle. Refer to Figure~\ref{fig:2dgadgets} for the exact placement.

We define the following gadgets using point sequences on this set of ten points
(see also Figure~\ref{fig:2dgadgets}).  Note that the $A$-gadgets traverse
the points clockwise, while the $B$-gadgets traverse the points
counter-clockwise.
\begin{eqnarray*}
g_A&:= p_{2,3}~p_{3,3}~p_{1,3} &\quad \widehat{g}_A := p_{1,3}~(g_A)^s \\
g_B&:= p_{1,3}~p_{3,3}~p_{2,3} &\quad \widehat{g}_B := p_{1,3}~(g_B)^s\\
g_a&:= p_{2,1}~p_{3,1}~p_{1,1} &\quad \widehat{g}_a := p_{1,1}~(g_a)^s\\
g_b&:= p_{3,1}~p_{2,1}~p_{1,1} &\quad \widehat{g}_b := p_{1,1}~(g_b)^s
\end{eqnarray*}

We now describe the mapping of an instance of the {\sc SCS} decision problem to our
instance space.  Given a string $s_i \in S$, we replace every
letter $A$ and every letter $B$  as follows:
\begin{eqnarray*}
A ~\rightarrow~ (\widehat{g}_a~ \widehat{g}_b)^t~ \widehat{g}_A~ (\widehat{g}_a~ \widehat{g}_b)^t\\
B ~\rightarrow~ (\widehat{g}_a~ \widehat{g}_b)^t~ \widehat{g}_B~ (\widehat{g}_a~ \widehat{g}_b)^t
\end{eqnarray*}
We obtain each $\gamma(s_i)$ by concatenating these point sequences along the sequence $s_i$.
To complete the reduction to the \klc[1] decision problem, where 
distances are measured under the discrete Fr\'echet distance, we set $\ell= 6t^2 + 9t$.

\begin{theorem}\label{thm:2dhardness:discrete}
The \klc problem under the discrete Fr\'echet distance 
    is $\mathsf{NP}$-hard to approximate within any factor smaller than $3\sin\frac{\pi}{3}$
for curves in the plane or higher dimensions, even if $k=1$.
\end{theorem}

\subsection{The Proof (Discrete Fr\'echet distance)}

In the following, we prove Theorem~\ref{thm:2dhardness:discrete}.

\begin{lemma}
\label{2d:hardness:1}
    For any true instance of the {\sc SCS} decision problem, there exists a center curve
of length at most $(3s+3)t$ and radius at most $r=1$ in our construction.
\end{lemma}

\begin{proof}
If the instance is true, then there exists a common superstring of length at
most $t$. We map this superstring to a point sequence as follows. 
We define two center gadgets as follows (see also
Figure~\ref{fig:2dgadgets}). Note that, again, the $A$-gadget traverses
the points clockwise, while the $B$-gadget traverses the points
counter-clockwise.
\begin{eqnarray*}
c_A := {p_{2,2}~p_{3,2}~p_{1,2}}  &\quad \widehat{g}^{*}_A&:= p_{1,2}~(c_A)^s \\
 c_B := {p_{3,2}~p_{2,2}~p_{1,2}} &\quad \widehat{g}^{*}_B&:= p_{1,2}~(c_B)^s  
\end{eqnarray*}
We replace every letter $A$ and every letter $B$ as follows (see also Figure~\ref{fig:2dgadgets}):
\begin{eqnarray*}
A~ \rightarrow~ p_0~ \widehat{g}^{*}_A~ p_0 \\
B~ \rightarrow~ p_0~ \widehat{g}^{*}_B~ p_0
\end{eqnarray*}

    Clearly, the resulting point sequence has length at most $(3s+3)t$.  We claim that the
resulting curve is a center curve of radius at most $1$.  Since the instance is
true, every input string is a substring of the common superstring. This implies
a valid matching between the point sequences. Consider traversing the sequences
from the beginning. Whenever we need to `skip' a letter of the superstring, we can 
match this letter to a corresponding buffer gadget. Whenever we need to `skip'
a sequence of buffer gadgets, we match them to the point $p_0$ on the point
sequence of the superstring.
\end{proof}

\begin{figure}\centering
\includegraphics[page=1,width=0.3\textwidth]{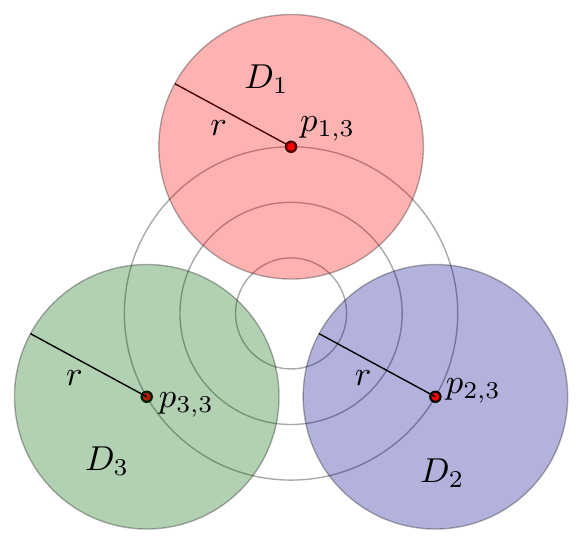}
~
\includegraphics[page=2,width=0.3\textwidth]{figures/2Dhardness_disks}
~
\includegraphics[page=3,width=0.3\textwidth]{figures/2Dhardness_disks}
\caption{Disks $D_1,D_2,D_3$ of radius $r$ centered at the points $p_{1,3}$,
$p_{2,3}$ and $p_{3,3}$. Any curve that lies within discrete Fr\'echet distance $r$ to
gadget $A$ or $B$ needs to have three vertices, one in each disk, where the ordering
distinguishes $A$ and $B$, as long as the disks are pairwise disjoint. 
Also shown: different approximation radii for the discrete and continuous
Fr\'echet distance.}
\label{fig:approx_radii}
\label{figs:threedisks}
\end{figure}

\begin{lemma}\label{lem:no:overlap:discrete}
Let $\psi$ and $\phi$ be two point sequences in the plane, such that 
$d_{DF}((g_A)^{s-1}, \psi) < r$ and $d_{DF}((g_B)^{s-1}, \phi) < r$ for
$r=3\sin\frac{\pi}{3}$ and for some value of $s > 1$. 
Each of the following statements is false:
\begin{compactenum}[(i)]
\item $\psi$ contains a prefix that is a suffix of $\phi$ 
\item $\phi$ contains a prefix that is a suffix of $\psi$ 
\item $\psi$ is a subsequence of $\phi$ 
\item $\phi$ is a subsequence of $\psi$ 
\end{compactenum}
\end{lemma}
\begin{proof}
Note that we chose the value of $r$ such that the three disks $D_1,D_2$ and
$D_3$ are disjoint, see also Figure~\ref{figs:threedisks}.  We prove falseness
of statement $(i)$ in detail. The other proofs are omitted since they are very
similar. The first point in the sequence $\psi$ lies in $D_2$, since
$(g_A)^{s-1}$ starts with $p_{2,3}$. The first point of $\psi$ that lies
outside of the disk $D_2$ must lie in the disk $D_3$, since the second vertex
of $(g_A)^{s-1}$ is $p_{3,3}$. This means that $\psi$ contains two consecutive
points $q$ and $q'$, the first inside $D_2$ and the second inside $D_3$. We
claim that $\phi$ does not contain $q$ and $q'$ consecutively. Indeed, there is
no point in $(g_B)^{s-1}$ that  would be within distance $r$ to both $q$ and
$q'$ and there are no two consecutive points in $(g_B)^{s-1}$ such that the
first is within distance $r$ to $q$ and the second is within distance $r$ to
$q'$. Therefore $(i)$ is false.
\end{proof}

\begin{lemma}
\label{2d:hardness:2}
    For any instance of an {\sc SCS} decision problem, consider the above reduction with $s= 2t+2$.  
If there exists a center curve of size at most $(3s+3)t$ and of radius
strictly smaller than $r=3\sin\frac{\pi}{3}$ then the instance is true.
\end{lemma}

\begin{proof}
Let the center curve be denoted by $\beta$. We claim that we can transform
    this curve into a common superstring of $S$---the instance of the  {\sc SCS} decision
problem.  We define three disks of radius $r$ centered at the points $p_{1,3}$,
$p_{2,3}$ and $p_{3,3}$.  We denote the disks with $D_1$, $D_2$ and $D_3$, see
Figure~\ref{figs:threedisks}.  By construction, the three disks are pairwise
disjoint.
In the first phase, we remove all vertices of $\beta$ that are not
contained in or on the boundary of any of the disks $D_1,D_2$ or $D_3$. Next, we snap every vertex in
$D_1$ to its center $p_{1,3}$ and similarly, we snap each vertex in $D_2$ to
its center $p_{2,3}$ and every vertex in $D_3$ to its center in $p_{3,3}$.  In
the third phase we replace consecutive copies of the same vertex in the
sequence with a single copy of this vertex. We denote the resulting point
sequence with $\widehat{\beta}$. From this curve we can now directly generate
a superstring by replacing 
\begin{eqnarray*}
(g_A)^{s-1} ~\rightarrow~ A \\
(g_B)^{s-1} ~\rightarrow~ B 
\end{eqnarray*}
We remove all remaining points from the sequences that are not matched by the
above rule.  By Lemma~\ref{lem:no:overlap:discrete} there cannot be any overlap
between subsequences that produce different letters. Therefore, the resulting
sequence has a unique ordering.
We claim that the resulting string is a common superstring of the set of strings $S$.
Since $\beta$ is a center curve of radius $r$, it is within Fr\'echet
distance $r$ to $\gamma(s_i)$ for all $i=1, \dots, n$.  By construction, every
letter $A$ of $s_i$ generates a subsequence $(g_A)^{s-1}$ in $\gamma(s_i)$ and every
letter $B$ of $s_i$ generates a subsequence $(g_B)^{s-1}$ in $\gamma(s_i)$. In order to
match each subsequence $g_A$, the curve $\beta$ needs to visit the
disks $D_1,D_2$ and $D_3$ in the correct order, and have a vertex in each disk.
By our transformation, a subcurve of $\beta$ that matches to a subsequence
$(g_A)^{s-1}$ is transformed into the letter $A$. Similarly, a subcurve that matches
to a subsequence $(g_B)^{s-1}$ of $\gamma(s_i)$ is transformed into the letter $B$.
Therefore, the resulting sequence is a supersequence of $s_i$ for all $i=1, \dots, n$.

The size of the generated supersequence is at most 
\[\left\lfloor \frac{(3s+3)t}{3(s-1)} \right\rfloor 
= \left\lfloor \frac{(3(s-1)+6)t}{3(s-1)}\right\rfloor 
= t + \left\lfloor \frac{2t}{s-1} \right\rfloor.\] 
    For $s \geq 2t+2$, this is at most $t$. This implies that the {\sc SCS} instance is true.
\end{proof}

From the above lemmas we obtain Theorem~\ref{thm:2dhardness:discrete}. Extending the analysis to the continuous Fr\'echet distance yields  Theorem~\ref{thm:2dhardness:continuous}. Refer to  Section~\ref{sec:2Dhard-app} for details.

\begin{restatable}{theorem}{hardtwodcontappx}
\label{thm:2dhardness:continuous}
The \klc problem under the continuous Fr\'echet distance 
    is $\mathsf{NP}$-hard to approximate within any factor smaller than $2.25$
for curves in the plane or higher dimensions, even if $k=1$.
\end{restatable}

%% file: hardness2D-appendix.tex

\subsection{Continuous Fr\'echet Distance}
\label{sec:2Dhard-app}

We prove that the reduction described in Section~\ref{sec:2dhard} implies hardness of approximation for any factor smaller than $2.25$ in the case where distances are measured under the continuous Fr\'echet distance.
(When concatenating two point sequences, we implicitly insert the edge connecting the endpoints.) 
Lemma~\ref{2d:hardness:1} still holds as is. Lemma~\ref{2d:hardness:2} holds
for the current construction only by using a slightly smaller radius, 
namely $r=2.25$.
The current construction would not work with a larger radius if distances are
measured under the continuous Fr\'echet distance. Intuitively, the center curve
would be able to match to an $A$-gadget or $B$-gadget using fewer vertices,
since the distance of an edge of the triangle $(p_{1,3}, p_{2,3}, p_{3,3})$ to
its opposite vertex is smaller than $2r$. In this case, the length of the
generated supersequence might exceed $t$. However, Lemma~\ref{2d:hardness:2:cont} below 
testifies that if we decrease the radius $r$ such that each entire edge of the
triangle is at distance greater $2r$ to its opposite vertex, then the reduction
still works. Refer to Figure~\ref{fig:approx_radii} for a visualization of the
used radii. 

\begin{figure}\centering
\includegraphics[width=0.9\textwidth]{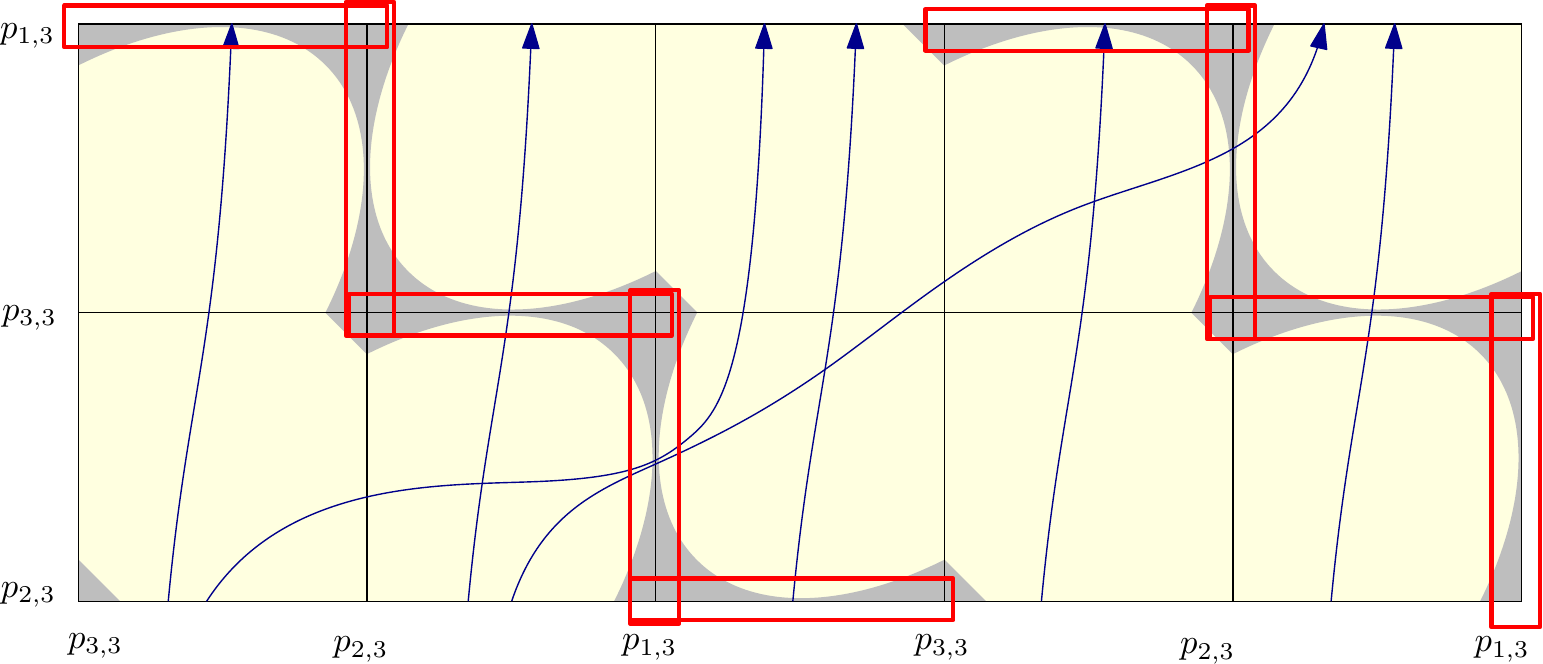}
\caption{$\delta$-Free-space diagram of the A-gadget $g_A$ and two cycles of the
B-gadget $g_B$. In this diagram, $\delta=2r$ for $r=2.25$. Red boxes indicate vertex-edge pairs
of a corner of the triangle with its opposite edge. These parts of the diagram are blocked 
for $\delta<2r$. Examples of paths spanning the full A-gadget and smaller and larger 
portions of the B-gadget are shown. The figure illustrates that such paths are not feasible for
$\delta<2r$. } \label{fig:2dgadgets:freespace}
\end{figure}

We start by proving the equivalence of Lemma~\ref{lem:no:overlap:discrete} in the
case that distances are measured under the continuous Fr\'echet distance.

\begin{lemma}\label{lem:no:overlap:count}
Let $\gamma:[0,1] \rightarrow \mathbb{R}^2$ be a polygonal curve in the plane
and let $t_1,t_2,t_3,t_4 \in [0,1]$.  If $d_F((g_A)^{s-1}, \gamma[t_1,t_2]) <
r$ and $d_F((g_B)^{s-1}, \gamma[t_3,t_4]) < r$, for $r=2.25$, then 
$[t_1,t_2]~\cap~[t_3,t_4]~=~\emptyset$.
\end{lemma}
\begin{proof}
For the sake of contradiction, assume that
$[t_1,t_2]~\cap~[t_3,t_4]~\neq~\emptyset$ and denote with $[t_5,t_6]$ the
interval of intersection. This implies that there exists a subcurve $\gamma_A$
of $g_A^{s-1}$ and a subcurve $\gamma_B$ of $g_B^{s-1}$ with the property that
$d_F( \gamma(t_5,t_6), \gamma_A) \leq r$ and $d_F( \gamma(t_5,t_6), \gamma_B)
\leq r$. Since the Fr\'echet distance satisfies the triangle inequality, this
implies $d_F(\gamma_A,\gamma_B) < 2r$. 
There are four different cases 
\begin{inparaenum}[(i)]
\item $t_5=t_1$ and $t_6=t_2$
\item $t_5=t_3$ and $t_6=t_4$
\item $t_5=t_3$ and $t_6=t_2$
\item $t_5=t_1$ and $t_6=t_4$
\end{inparaenum}
In the first case $\gamma_A$ is equal to the entire curve $g_A^{s-1}$.
Figure~\ref{fig:2dgadgets:freespace} illustrates that in this case, the
Fr\'echet distance $d_F(\gamma_A,\gamma_B)$ must be at least $2r$, which
contradicts the above.  In the second case, we have the symmetric case that
$\gamma_B$ is equal to the entire curve $g_B^{s-1}$.  In the third case,
$\gamma_A$ is a suffix curve of $g_A^{s-1}$ and $\gamma_B$ is a prefix curve of
$g_B^{s-1}$. In this case, there must be a monotone path in the free space 
diagram shown in Figure~\ref{fig:2dgadgets:freespace} that starts at a vertical 
line of a vertex at $p_{3,3}$ and ends at a horizontal line of a vertex at 
$p_{1,3}$. However, as the figure illustrates, this is only possible in case
$d_F(\gamma_A,\gamma_B) > 2r$. The fourth case is symmetric to the third. 
The claim follows by contradiction in all cases.
\end{proof}

\begin{lemma}\label{lem:size:count}
Let $\gamma:[0,1] \rightarrow \mathbb{R}^2$ be a polygonal curve in the plane and let $t_1,t_2 \in [0,1]$.
If $d_F((g_A)^{s-1}, \gamma[t_1,t_2]) < r$ or $d_F((g_B)^{s-1}, \gamma[t_1,t_2]) < r$, for $r=2.25$, then the number of vertices of
$\gamma[t_1,t_2]$ is at least $3(s-2)$ (excluding endpoints).
\end{lemma}
\begin{proof}
Assume for now that $\gamma$ matches to $g_A^{s-1}$ within Fr\'echet distance $r$. We can prove the claim by induction on the vertices of $\gamma$. $\gamma$ has to start in $D_2$, pass through $D_3$ before entering $D_1$. The three disks of radius strictly smaller than $r=2.25$ cannot be stabbed by a line. Therefore, $\gamma$ needs to have an extra vertex before entering $D_1$. Now we can make the same argument for the sequence of disks $D_3, D_1, D_2$--- there needs to be another vertex on $\gamma$ before the curve enters $D_2$. For each induction step we shift the window of three disks by one disk along the sequence $g_A^{s-1}$. This proves the claim for the $A$-gadget. The proof for the $B$-gadget is analogous.
\end{proof}

\begin{lemma}\label{lem:double:size:count}
Let $\phi:[0,1] \rightarrow \mathbb{R}^2$ be a polygonal curve in the plane,
such that either \begin{inparaenum}[(i)]
\item $d_F(\phi, \gamma(A~A)) < r$,
\item $d_F(\phi, \gamma(B~B)) < r$,
\item $d_F(\phi, \gamma(A~B)) < r$, or
\item $d_F(\phi, \gamma(B~A)) < r$
\end{inparaenum} for
$r=2.25$, then $\phi$ contains at least $6s$ points (excluding endpoints).  
\end{lemma}
\begin{proof}The proof follows along the lines of the proof of Lemma~\ref{lem:size:count}.
\end{proof}

\begin{lemma}
\label{2d:hardness:2:cont}
    For any instance of an {\sc SCS} decision problem, consider the above reduction with
$s= 3t+3$.  Assume that distances are measured under the continuous Fr\'echet
distance.  If there exists a center curve of size at most $(3s+3)t$ and of
radius strictly smaller than $r=2.25$ then the instance is true.
\end{lemma}

\begin{proof}
Let the center curve be denoted by $\beta$. 
By Lemma~\ref{lem:no:overlap:count}, subcurves of the center curve matching to
two different gadgets are disjoint.  Furthermore, by Lemma~\ref{lem:size:count}
the number of vertices of $\beta[t_1,t_2]$ is at least $3(s-2)$.

We can now scan $\beta$ from the beginning and generate a superstring of
    the {\sc SCS} instance as follows. As before, we define three disks of radius $r$
centered at the points $p_{1,3}$, $p_{2,3}$ and $p_{3,3}$.  We denote the disks
with $D_1$, $D_2$ and $D_3$, see Figure~\ref{figs:threedisks}.  By
construction, the three disks are pairwise disjoint. We find a sequence of
values $t_0,t_1,\dots$ as follows. We initialize $t_0=0$.  For $i>0$, let $t_i
\in [t_{i-1},1]$ be the minimal value of $t$, such that $\beta(t)$ is an
intersection with one of the disks $D_2$ or $D_3$.
If it is $D_2$, then we find the value of $t_{i+1}$, such that 
$d_F(\beta[t_{i},t_{i+1}], {(g_A)^{s-1}}) \leq r$ and we output ``A''. 
If it is $D_3$, then we find the value of $t_{i+1}$, such that 
$d_F(\beta[t_{i},t_{i+1}], {(g_B)^{s-1}}) \leq r$ and we output ``B''.
If there is no such $t_{i+1}$, then we set $t_{i+1}$ to be the first point
where the curve leaves the current disk and do not output anything.  We
continue this process until we reach the end of $\beta$. The sequence of
letters output this way is the generated supersequence.
By the above claims, the length of
the generated sequence is at most 
\[\left\lfloor \frac{(3s+3)t}{3(s-2)} \right\rfloor 
= \left\lfloor \frac{(3(s-2)+9)t}{3(s-2)}\right\rfloor 
= t + \left\lfloor \frac{3t}{s-2} \right\rfloor.\] 
For $s \geq 3t+3$, this is at most $t$.
\end{proof}

Finally, we obtain the following theorem.

\begin{theorem}\label{thm:2dhardness:continuous-app}
The \klc problem under the continuous Fr\'echet
    distance is $\mathsf{NP}$-hard to approximate within any factor smaller than $2.25$ for
curves in the plane or higher dimensions, even if $k=1$.
\end{theorem}

%% file: hardness-MEB.tex
\section{Extension to the Minimum Enclosing Ball Problem}\label{sec:meb}
In the previous sections we described variants of a reduction from the shortest common supersequence problem to the problem of computing a low-complexity curve that lies within a bounded radius from a set of input curves. In this section we show that when relaxing the bound on the complexity of the center curve the problem remains just as hard. The unconstrained center curve can be interpreted as the solution to the decision version of the problem of finding the minimum enclosing ball in the corresponding metric space of curves. 


In the presentation, we focus on the case of the discrete Fr\'echet distance on one-dimensional curves. 
We again reduce from the {\sc Shortest Common Supersequence} ({\sc SCS}) problem. Given an instance of the decision version of the {\sc SCS} problem, i.e. a set strings $S=\{s_1,s_2,\ldots,s_n\}$ over an alphabet $\{A,B\}$ and the maximum allowed length $t$ for the supersequence sought for the strings in $S$, we construct a corresponding instance of the decision version of the {\sc Minimum Enclosing Ball} ({\sc MEB}) problem for discrete Fr\'echet distance with $\delta=1$ and set of curves $G\cup R_{j,j'}$. Here, $G$ is the set of curves $\gamma(s_i)$ constructed from the input sequences $s_i\in S$, as used in the construction from Section~\ref{sec:1d}. $R_{j,j'} = \{A^j,B^{j'}\}$, with 
\[
\begin{aligned}
 A^j &= p_1(p_{-3}p_1)^j\\
 B^{j'} &= p_{-1}(p_3p_{-1})^{j'}.
\end{aligned}
\]
We will show that the instance $(S,t)$ is a true instance of {\sc SCS} if and only if there exists a pair $(j,j')\in I_t :=\{(j,j')\in\mathbb{N}^2\mid  j+j'=t\}$ such that $(G\cup R_{j,j'},1)$ is a true instance of MEB for discrete Fr\'echet distance. Since the number of choices for pairs $(j,j')$ is linear in $t$, the number of different instances of MEB is polynomial in the input size of the SCS instance. (Although the input size of $t$ is $\log{t}$, SCS on binary alphabets is non-trivial only when $t\leq 2\max_{s_i\in S} |s_i|$ so we can assume that $t$ is at most linear in the input size of the SCS instance). This means that there is a polynomial time truth-table reduction from SCS to MEB for discrete Fr\'echet distance.

\begin{lemma}\label{lem:SCS-implies-j-k-center-curve}
If $(S,t)$ is a true instance of SCS, then there exists $(j,j')\in I_t$ such that $(G\cup R_{j,j'}, 1)$ is a true instance of MEB under the discrete Fr\'echet distance.
\end{lemma}
\begin{proof}
If $(S,t)$ is a true instance of SCS, then there is a string $s^*$ of length at most $t$ that is a supersequence of all strings in $S$. Without loss of generality, we may assume that $s^*$ has length exactly $t$, as we can always pad a shorter supersequence with arbitrary characters. 

Construct the curve $c^*$ of length $2t+1$, with for each $i\in \{1,\ldots, 2t+1\}$, 
\[c^*_i := \begin{cases} p_0 & \text{$i$ is odd}\\ p_{-2} & \text{$i$ is even and $s^*_{i/2}=A$ } \\ p_{2} & \text{$i$ is even and $s^*_{i/2}=B$.} \end{cases}\]

That is, $c^*$ has a vertex at $p_{-2}$ for every $A$ in $s^*$, a vertex at $p_{2}$ for every $B$ in $s^*$ and with $p_0$ in between and $p_0$ as the first and last vertex of $c^*$. Note that this is exactly the same curve as used in Section~\ref {sec:1d}.

Let $\gamma(s_i)\in G$. We will create an alignment with $c^*$. Since $s_i$ is a subsequence of $s^*$, every letter gadget in $\gamma(s_i)$ can be matched with the corresponding letter gadget in $c^*$ within distance $1$. The remaining $p_2$ and $p_{-2}$ vertices in $c^*$ can be matched with a vertex $p_1$ and $p_{-1}$ in a buffer of $\gamma(s_i)$, respectively. All other vertices in the buffer gadgets of $\gamma(s_i)$ can be matched with some $p_0$ in $c^*$. We now have an alignment between $\gamma(s_i)$ and $c^*$ where all matched vertices have distance at most $1$, so $d_{DF}(\gamma(s_i),c^*)\leq 1$ for all $i$. 

Set $j$ to the number of occurrences of $A$ in $s^*$ and $j'$ to the number of occurrences of $B$ in $s^*$. Then, $j+j'=t$, so $(j,j')\in I_t$. Since $c^*$ contains a vertex at $p_{-2}$ exactly $j$ times, we can match those characters to $p_{-3}$ in $A^j$ and all remaining vertices in $c^*$, at $p_0$ or $p_{2}$, with $p_1$. 
Since all matched pairs have a distance of $1$, the discrete Fr\'echet distance between $A^j$ and $c^*$ is at most $1$. 
Analogously, we get that the discrete Fr\'echet distance between $B^{j'}$ and $c^*$ is at most $1$. 
So, $d_{DF}(g,c^*)\leq 1$ for all $g\in G\cup R_{j,j'}$ for some $(j,j')\in I_t$.
\end{proof}

\begin{lemma}\label{lem:j-k-center-curve-implies-SCS}
If there exists $(j,j')\in I_t$ such that $(G\cup R_{j,j'},1)$ is a true instance of MEB under the discrete Fréchet distance, then $(S,t)$ is a true instance of SCS.
\end{lemma}
\begin{proof}
If there exists $j,j' \in I_t$ such that $(G\cup R_{j,j'},1)$ is a YES-instance, then there exist a center curve $c^*$ such that $d_{DF}(g,c^*)\leq 1$ for all $g\in G\cup R_{j,j'}$. 
So, there exists a alignment between $c^*$ and $A^j$ in which all matched pairs have distance at most $1$. 
Note that there exists no vertex $p$ within distance $1$ to both $p_1$ and $p_{-3}$. 
This means that for every vertex in $c^*$, it is either matched to some vertices at $p_{-3}$ or some vertices at $p_1$. 
Since every pair of vertices at $p_{-3}$ in $A^j$ has at least one vertex at $p_1$ in between and an alignment is monotone, a point in $c^*$ matched to some points $p_{-3}$ can be matched to at most one such point in $A^j$. 
The same holds for the points matched to $p_1$. So, every point in $c^*$ is matched to exactly one point in $A^j$. 

\begin{figure}[ht]
\centering
\includegraphics{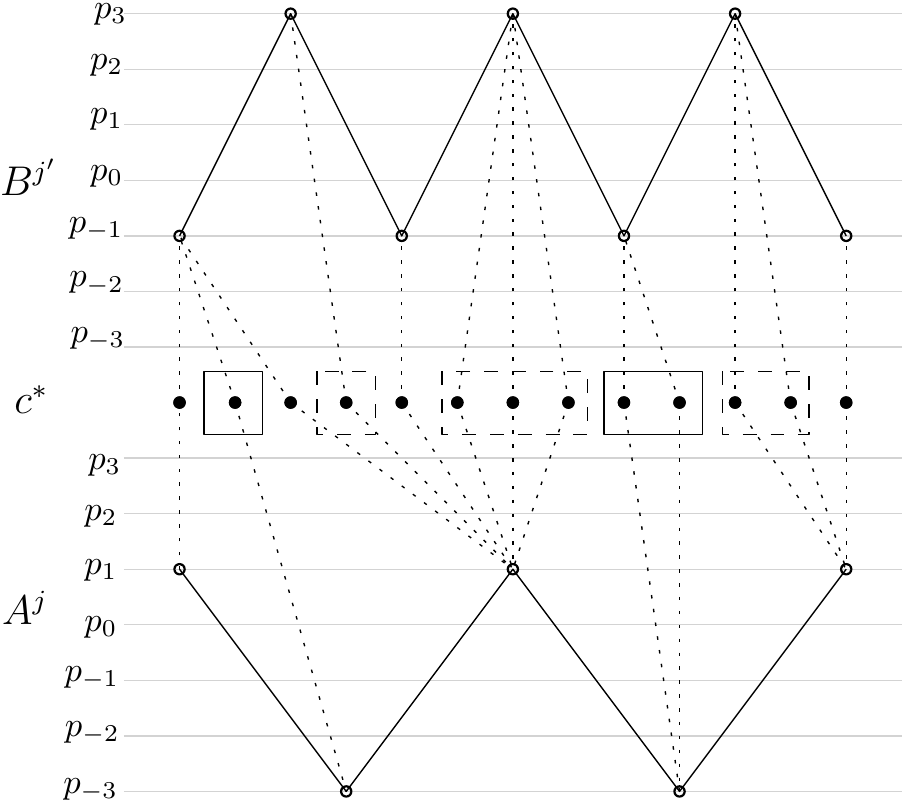}
\caption{ The partition of a center curve $c^*$ from Lemma~\ref{lem:j-k-center-curve-implies-SCS}, based on matchings between the center curve $c^*$ and curves $A^j$ and $B^{j'}$. Dashed boxes indicate an A-part, normal boxes indicate a B-part, and the points without boxes are in a buffer part.}
\label{fig:MEB-center-partition}
\end{figure}

This means we can partition the points of $c^*$ into $2j+1$ parts, where each part is represented by a point on $A^j$ such that all points of $c^*$ in that part are matched to that point in $A^j$.
We can do the same with the matching between $c^*$ and $B^{j'}$ to get a partition of $2j'+1$ parts. 
A point in $c^*$ cannot be both in a part corresponding to an $p_{-3}$ in $A^j$ and in another corresponding to an $p_{3}$ in $B^{j'}$.
So, we can `combine' these two partitions into one partition of at most $2(j+j')+1$ parts (see Figure~\ref{fig:MEB-center-partition}), as follows: 
if a point in $c^*$ is matched to $p_{-3}$, put it in the part corresponding to that point, we call this an A-part. 
If a point in $c^*$ is matched to $p_{3}$, put it in the part corresponding to that point, we call such a part an B-part. 
If a point is matched to neither $p_{-3}$ nor $p_3$, put it in a part corresponding to the `buffer' between the part of the nearest point of smaller index in $c^*$ in an A-or B-part and the nearest point of greater index in an A-or B-part. 
(If such an index does not exist, identify it by only one A-or B-part. There are at most two such parts, at the start and end of $c^*$.)
We call such a part a buffer-part. 
Note that by construction, there are exactly $j+j'=t$ A-and B-parts and at most $j+j'+1$ buffer-parts. 

Construct the string $s^*$ by removing the buffer parts from the partition of $c^*$,  replacing the A-parts with the $A$ character and the B-parts with $B$ character. As there are exactly $j+j'=t$ A-and B-parts, $|s^*|=t$. We complete the proof by showing that $s^*$ is a supersequence of all sequences in $S$.

Let $\gamma(s)\in G$. As $d_{DF}(\gamma(s),c^*)\leq 1$, there exists a matching between $\gamma(s)$ and $c^*$ such that all matched pairs have distance at most $1$. This means an A-part cannot be matched to multiple letter gadgets, since the matching is monotone and the buffer separating the letter gadgets contains the point $p_1$. Analogously, the B-parts cannot be matched to multiple letter gadgets. Furthermore, each A letter gadget must be matched to an A-part and each B letter gadget to a  B-part. This means that the sequence of letter gadgets in $\gamma(s)$ is matched one-to-one with a subsequence of the sequence of A-and B-parts in $c^*$ with matching letters and therefore $s$ is a subsequence of $s^*$.
\end{proof}

So, we have a polynomial time reduction from SCS to MSB under discrete Fr\'echet distance and we can conclude with the following theorem.

\begin{theorem}
The minimum enclosing ball problem for polygonal curves in 1D for the discrete Fréchet distance is $\mathsf{NP}$-hard.
\end{theorem}

The same construction can be generalised to the continuous case. We state the theorem here. Refer to the Appendix~\ref{sec:meb-app} for the full details and proofs.

\begin{restatable}{theorem}{hardonedMEBcont}
\label{thm:MEB:cont}
The minimum enclosing ball problem for polygonal curves in 1D for the continuous Fréchet distance is $\mathsf{NP}$-hard.
\end{restatable}

%% file: approx.tex
\section{Clustering Algorithms}\label{sec:clustering}

We sketched the basic algorithm in Section~\ref{sec:basic:algorithm:sketch}. In this section we complete the picture by providing the remaining details, including the specifics of the simplification, and a full analysis of the running time. In Section~\ref{sec:imp-apx-cl} we will show how to improve the approximation factor and, as such, obtain our main algorithmic result.

Recall that, given a polygonal curve $\gamma$, a minimum-error $\ell$-simplification $\widehat{\gamma}$ is a curve of complexity at most $\ell$ whose error $\eps$ is minimal.
We define a \emph{$c$-approximate $\ell$-simplification} $\overline{\gamma}$ as a curve of complexity at most $\ell$ whose error is at most a factor of $c$ from the optimal, i.e. $d(\gamma,\overline{\gamma}) \leq c\eps$. Under the discrete \fd Bereg~et~al.~\cite{bjwyz-spcdfd-08} presented an $O(\ell m \log{m} \log{(m/\ell))}$-time algorithm that computes a minimum-error $\ell$-simplification. 
For the continuous \fd, we briefly discuss an algorithm to compute a $4$-approximate $\ell$-simplification, before stating our algorithm for the \klc problem.

\subsection{Computing an \texorpdfstring{$\ell$}{l}-simplification for the Continuous Fr\'{e}chet Distance}
Agarwal~et~al.~\cite{ahmw-nltaa-05} give the following theorem, where $\kappa_{F}(\eps,\gamma)$ and $\kappa_{\hat{F}}(\eps,\gamma)$ denote the complexity of the vertex-constrained and weak minimum-complexity $\eps$-simplification respectively, under the continuous \fd.

\begin{theorem}[Agarwal~et~al.~\cite{ahmw-nltaa-05}, Theorem 4.1] \label{thm:4-apx}
    Given a polygonal curve $\gamma$, \[\kappa_{F}(\eps,\gamma) \leq \kappa_{\hat{F}}(\eps/4,\gamma).\]
\end{theorem}

This result gives the following algorithm for approximating the weak minimum-error \ls.

\begin{lemma} \label{lem:4apxLS}
    Given a curve $\gamma$ of $m$ vertices, a $4$-approximate $\ell$-simplification under the continuous \fd can be computed in $O(m^3\log{m})$ time.
\end{lemma}
\begin{proof}
    Theorem~\ref{thm:4-apx} implies that, given a \emph{weak} minimum-error \ls $\overline{\gamma}$ of $\gamma$ with error $\eps$, there exists a \emph{vertex-constrained} \ls $\overline{\gamma_v}$ 
    with complexity at most $\ell$ and error at most $4\eps$.
    The algorithm by Imai and Iri can be run to compute the simplification, 
    by computing the \fd between a segment and a curve in $O(m\log{m})$ using the algorithm by Alt and Godau~\cite{ag-cfdbt-95}.
    This gives an $O(m^3\log{m})$ algorithm to compute the minimum-error vertex-constrained \ls, which is the required $4$-approximate  \ls.
\end{proof}


\subsection{The Basic Algorithm}\label{sec:basic:algorithm}
Now we are ready to state our adaption of Gonzalez' algorithm.  
A set $\mathcal{G}$ of $n$ polygonal curves is given as input.  
In the first iteration we choose an arbitrary curve $\gamma_1$ from $\mathcal{G}$ and we add a $c$-approximate $\ell$-simplification $\overline{\gamma_1}$ of $\gamma_1$ to the initially empty set of centers $\mathcal{C}$. 
The algorithm now iteratively computes a
sequence of centers.  In the $i$-th iteration, the algorithm finds the curve $\gamma_i \in \mathcal{G}$ that maximizes the quantity
\begin{equation*}
d_{i-1}(\gamma_i) =  \min_{1\leq j \leq i-1} d(\overline{\gamma_{j}}, \gamma_i).
\end{equation*}
This curve is ``farthest'' from the current set of centers. The algorithm then computes a $c$-approximate $\ell$-simplification $\overline{\gamma_i}$ of
$\gamma_i$ and adds it to $\mathcal{C}$. After $k$ iterations the algorithm returns the set $\mathcal{C}=\{\overline{\gamma_1},\dots,
\overline{\gamma_k}\}$. 

\begin{theorem} \label{thm:c2:approx}
    Given $n$ polygonal input curves, each of complexity $m$, and positive integers $k, \ell$, the above algorithm computes a $(c+2)$-approximation to the \klc problem  in time $O(kn \cdot \ell m \log(\ell+m)+ k\cdot T_{\ell}(m))$, where $T_{\ell}(m)$ denotes the time to compute a $c$-approximate $\ell$-simplification of a polygonal curve of complexity $m$.
\end{theorem}
\begin{proof}
We first consider the running time. In each of the $k$ iterations the algorithm performs less than $n$ Fr\'echet distance computations between curves of complexity $m$ and $\ell$, followed by a single call to the $c$-approximate $\ell$-simplification algorithm with an input curve of complexity $m$ as parameter. Using the algorithm by Alt and Godau~\cite{ag-cfdbt-95} for each of these distance computations immediately gives the stated running time.

Next we prove the approximation bound. We claim that the set of centers $\mathcal{C}=\{\overline{\gamma_1},\dots,
\overline{\gamma_k}\}$ computed by the algorithm is a $(c+2)$-approximation to the \klc of $\mathcal{G}$. To see this, we adapt the proof by Har-Peled~\cite{h-gaa-11} (see also Gonzalez~\cite{g-cmmid-85}). Recall that in the $i$-th iteration of the algorithm we maintain for each $\gamma \in \mathcal{G}$ its distance to the current clustering $d_{i-1}(\gamma)$. 
We also maintain the radius of the clustering
\begin{equation*}
r_{i-1} = \max_{\gamma \in \mathcal{G}} d_{i-1}(\gamma) =  \max_{\gamma \in \mathcal{G}} \min_{1\leq j\leq i-1} d(\overline{\gamma_{j}}, \gamma).
\end{equation*}
By construction, we have that $r_1 \geq r_2 \geq \dots \geq r_k$.
We simulate performing another iteration by
    selecting the curve $\gamma_{k+1} \in \mathcal{G}$ that realizes $r_{k}$.
From the above discussion, we have 
\begin{equation}\label{pair:lb}
\forall~ 1\leq i < j  \leq k+1~:\quad 
d(\overline{\gamma_i}, \gamma_j) 
~\geq~ 
r_k.
\end{equation} 

Now, consider the decomposition of $\mathcal{G}$ into clusters induced by an  optimal solution $\mathcal{C^*}$ having cost $r^*$. That is, every $\beta_i \in \mathcal{C^*}$ defines a subset of $\mathcal{G}$ containing the curves $\alpha \in \mathcal{G}$ such that $\beta_i$ is the curve that minimizes $d(\alpha,\beta_i)$ over all curves in $\mathcal{C^*}$.
By the pigeon hole principle, there must be two curves of $\gamma_1,\dots,\gamma_{k+1}$
that lie in the same cluster. Let $\gamma_i$ and $\gamma_j$ denote those curves
and assume without loss of generality that $i<j$. Finally, let $\beta \in \mathcal{C^*}$ be the corresponding cluster center.  
Using \eqref{pair:lb} and the fact that the \fd satisfies the triangle inequality, 
we have 
\begin{eqnarray*}
r_k \leq d(\overline{\gamma_{i}}, \gamma_j) 
\leq d(\overline{\gamma_{i}}, \gamma_i) + d(\gamma_i, \beta) + d(\beta, \gamma_j) 
\leq (c+2) r^*
\end{eqnarray*}

The last inequality follows from two facts. First we have that $\gamma_i$ and
$\gamma_j$ are contained in the optimal cluster centered at $\beta$ and thus
their distance to $\beta$ is upper bounded by $r^*$. Secondly, we have that
\[ d(\overline{\gamma_i}, \gamma_{i}) \leq c \cdot d(\gamma_{i},\beta),\]
    since $\overline{\gamma_i}$ is a $c$-approximation of the minimum-error $\ell$-simplification of $\gamma_i$, and $\beta$ is a curve of complexity $\ell$ so must be at least $\eps$ from $\gamma_i$. 
\end{proof}


The following result is implied by using the minimum-error $\ell$-simplification algorithm under the discrete \fd by Bereg~et~al.~\cite{bjwyz-spcdfd-08} or the algorithm described in Lemma~\ref{lem:4apxLS} for the continuous \fd.

\begin{corollary} \label{cor:alg:basic}
    Given $n$ polygonal input curves in $\R^d$, each of complexity $m$, and positive integers $k,\ell$, the above algorithm can compute a $3$-approximation to the \klc problem using the discrete \fd in time $O(kn \cdot \ell m \log (\ell+m))$, or a $6$-approximation under the continuous \fd in time $O(km(n \ell \log(\ell+m)+ m^2\log{m}))$.
\end{corollary}



%% file: approx-appendix.tex
\section{Improving the Approximation Factor} \label{sec:imp-apx-cl}
In this section we improve the approximation factor with respect to the basic algorithm described in Section~\ref{sec:basic:algorithm:sketch} by using an algorithm to compute the minimum-complexity $\eps$-simplification 
as an \emph{approximate decider} in a search for the approximately optimal $\ell$-simplification. 
%
%
Such an algorithm can then be used to compute a bicriteria approximation algorithm for the \klc problem. 
In the case of the continuous \fd in 2D, this technique 
improves the approximation factor to $3$ by utilizing an algorithm by Guibas et al.~\cite{ghms-91} to compute a minimum-complexity $\eps$-simplification in $O(m^2 \log^2 m)$ time.



First, observe that a similar statement to Theorem~\ref{thm:c2:approx} can be proven to obtain an approximate decision algorithm for the clustering problem.
That is, given a set $\mathcal{G}$ of $n$ polygonal curves and a parameter $\delta$ as input, run the curve clustering algorithm 
described in the previous subsection, with the modification that---when computing the simplification in each iteration---a minimum-complexity $\delta$-simplification is computed instead of an approximate minimum-error $\ell$-simplification. 
If the complexity of the resulting simplified curve is larger than $\ell$, then no $\ell$-simplification exists with error at most $\delta$, and the algorithm aborts and returns that there is no solution to the \klc instance with cost at most $\delta$. 
If the algorithm does not abort in any of the $k$ iterations, then it returns the computed set of centers. 

\begin{lemma} \label{lem:3:approx:dec}
    Given $n$ input curves, each of complexity $m$, and a decision parameter $\delta$, there exists an $O(kn \cdot \ell m \log(\ell+m)+ k \cdot T_{\delta}(m))$-time approximate decision algorithm for the \klc problem that answers one of the following:
        \begin{compactenum}[(i)]
        \item $\Delta^* \leq 3\delta$, or
        \item $\Delta^* > \delta$,
        \end{compactenum}
    where $T_{\delta}(m)$ denotes the time to compute a minimum-complexity $\delta$-simplification of a curve, and $\Delta^*$ is the cost of an optimal solution.  
        The answer is correct in both cases, and in the first case, a clustering is obtained with cost at most $3\delta$.
\end{lemma}
\begin{proof}
The running time is the same as in Theorem~\ref{thm:c2:approx}, and thus it remains to prove the approximation bound. 
We adapt the proof of Theorem~\ref{thm:c2:approx}, such that, in the case
where a solution $\mathcal{C'}$ with cost at most $\delta$ exists, 
    each minimum-complexity $\delta$-simplification of an input curve 
    must have complexity at most
$\ell$. Indeed, the center curves of $\mathcal{C'}$ have complexity at most $\ell$ and each input
curve has Fr\'echet distance at most $\delta$ to at least one of them.
Therefore, if $\Delta^* \leq \delta$, then the algorithm must complete and return a
solution without aborting. We claim that in this case the returned solution has cost at most
$3\delta$.  
    Recall from 
    Theorem~\ref{thm:c2:approx} that 
\begin{eqnarray*}
r_k \leq d(\overline{\gamma_{i}}, \gamma_j) 
\leq d(\overline{\gamma_{i}}, \gamma_i) + d(\gamma_i, \beta) + d(\beta, \gamma_j),
\end{eqnarray*}
holds for the cost of the clustering $r_k$ computed by the algorithm.
Observe that $d(\gamma_i, \beta) \leq \Delta^* \leq \delta$ and 
$d(\beta, \gamma_j) \leq \Delta^* \leq \delta$, given that $\beta$ is a center in the 
solution. By construction $d(\overline{\gamma_{i}}, \gamma_i)\leq \delta$ since 
$\overline{\gamma_{i}}$ is a minimum-complexity $\delta$-simplification of $\gamma$. 
The claim follows.
\end{proof}


This approximate decision algorithm can then be used in an algorithm to search for an improved approximation of the solution to the \klc problem.

\begin{theorem}\label{thm:approx:search}
    Given a $c_1$-approximation algorithm $\mathcal{A}_1$ for the \klc problem and a $c_2$-approximate decision algorithm $\mathcal{A}_2$ for the \klc problem, a $c_2$-approximate solution to the \klc problem can be computed with one call to $\mathcal{A}_1$ and $O(\log(c_1)/\log(c_2))$ calls to $\mathcal{A}_2$. 
\end{theorem}
\begin{proof}
Let $\alpha$ be the real cost of the approximate solution returned by algorithm $\mathcal{A}_1$.  Since $\mathcal{A}_1$ computes a $c_1$-approximate solution, this implies that $\Delta^* \in [\alpha/c_1, \alpha]$. We now call
algorithm $\mathcal{A}_2$ with $\delta=3^i \alpha/c_1$ for integer
$i=0,1,\dots$ until it returns a $c_2$-approximate solution.  This must happen for some $i$ with $c_2^i \alpha/c_1 \leq \alpha$. The theorem now follows.
\end{proof}


Let $\mathcal{A}_1$ be an instance of the algorithm described in Theorem~\ref{thm:c2:approx} that uses the algorithm from Lemma~\ref{lem:4apxLS} as a subroutine to compute a $4$-approximate \ls.
Similarly, let $\mathcal{A}_2$ be an instance of the algorithm from Lemma~\ref{lem:3:approx:dec} that calls the $O(m^2\log^2 m)$-time minimum-complexity simplification algorithm by Guibas et al.~\cite{ghms-91} for planar curves.
Theorem~\ref{thm:approx:search} thus implies the following:


\improvedalgresult*


%% file: acknowledgements.tex
\section*{Acknowledgements}
Natasja van de l'Isle is working on practical algorithms for curve clustering. We would like to thank her for sharing the example from Figure~\ref{fig:pigeon-example} with us.

%% file: hardness1D-appendix.tex
\section{Hardness Results in 1D}\label{sec:1d-app}

In this section we show that the decision version of the \klc[1] problem for discrete Fr\'echet distance as well as continuous Fr\'echet distance is $\mathsf{NP}$-hard in one dimension.  Later, in Section~\ref{sec:1Dappx-app} we will show that it is also hard to approximate the optimal distance within a constant factor. In Section~\ref{sec:meb-app} we show that the construction for the $\mathsf{NP}$-hardness of the minimum enclosing ball problem extends to the continuous Fr\'echet distance.

We reduce from the {\sc Shortest Common Supersequence} ({\sc SCS}) problem~\cite{RU81}, which asks to compute a shortest sequence $s^*$ such that each of $n$ finite input strings $s_i$ over an alphabet $\Sigma=\{A,B\}$ is a subsequence of $s^*$.

Given an instance of the decision version of the {\sc SCS} problem, i.e., a set of strings $S=\{s_1,s_2,\dots,s_n\}$ over an alphabet $\{A,B\}$, and a maximum allowed length $t$ of the sought supersequence, we will construct a corresponding instance of \klc[1] problem for Fr\'echet distance $\delta=1$. Let $p_{i}=(i)$ for $-3 \le i \le 3$ be seven points with the respective integer coordinates.

\subsection{Discrete Fr\'echet Distance}\label{sec:1Dhard-discrete-app}
For each input string $s_i\in S$ we construct a one-dimensional curve $\gamma(s_i)$ in the following way. For each character of $s_i$ we build one of the two \emph{letter} gadgets, which we connect into a curve with \emph{buffer} gadgets. Each curve will also begin and end with a buffer gadget. Specifically, consider the following point sequences (each of size $1$):
\[
g_A:= p_{-3} \,, \quad g_B:= p_{3} \,, \quad g_a:= p_{-1}\,, \quad g_b:= p_{1}\,.
\]
We map each letter of $s_{i}$ to a one-dimensional piece of curve in the following way:
\[
\begin{aligned}
A &\rightarrow~ (g_{a}~ g_{b})^{t}~ g_A~ (g_{b}~ g_{a})^{t} \,,\\
B &\rightarrow~ (g_{b}~ g_{a})^{t}~ g_B~ (g_{a}~ g_{b})^{t} \,.
\end{aligned}
\]
We obtain curve $\gamma(s_{i})$ by concatenating the resulting curves. We call $g_{A}$ and $g_{B}$ the letter gadgets, and the subcurves in between---the buffer gadgets. That is, a buffer gadget at the end of the curve consists of $(g_{a}~ g_{b})^{t}$ or $(g_{a}~ g_{b})^{t}$, and an intermediate buffer gadget consists of a combination of two of these sequences. Figure~\ref{fig:1dgadgetsD-example} shows an example of three curves constructed for strings $ABB$, $BBA$, and $ABA$.


Consider the question of finding a curve $c_i$ of minimum complexity within discrete Fr\'echet distance $\delta=1$ of a curve $\gamma(s_i)$ constructed for some string $s_i\in S$. One vertex of $c_i$ cannot completely cover a letter gadget and a neighboring buffer gadget, that is, it cannot be within distance $\delta=1$ from all the points of the two gadgets. Furthermore, a buffer gadget cannot be covered by two consecutive vertices which are also covering the neighboring letter gadgets. Therefore, the curve $c_i$ must have at least one vertex per gadget. Thus, $|c_i|\ge 2|s_i|+1$. 

Now, consider a curve $c'_i$ that has a vertex at $p_{0}$ for every buffer gadget, a vertex at $p_{-2}$ for every $A$-gadget, and a vertex at $p_{2}$ for every $B$-gadget (refer to Figure~\ref{fig:1dgadgetsD-example}). Then, the discrete Fr\'echet distance between $\gamma(s_i)$ and $c'_i$ is $1$. Note, that $|c'_{i}|=2|s_i|+1$. On the other hand, one buffer gadget can completely cover $c'_i$, i.e., $d_{DF}(\text{buffer gadget},c'_i)=1$. Indeed, every vertex of $c'_i$ is within distance $1$ from one of the `spikes' of the buffer gadget (points $p_{-1}$ or $p_{1}$), and the number of such spikes in a buffer gadget is larger than the number of points of $c'_i$ corresponding to a letter.

Having constructed a curve $\gamma(s_i)$ for each string $s_i$, we will now show that there exists a string $s^*$ of length $t$ that is a supersequence of all strings in $S$ if and only if there exists a center curve $c^*$ with $\ell=2t+1$ vertices that lies within discrete Fr\'echet distance $\delta=1$ from all curves $\gamma(s_i)$.

\begin{lemma}
    For any true instance of the {\sc SCS} problem, there exists a center curve of length at most $2t+1$.
\end{lemma}
\begin{proof}
Let $s^*$ be a supersequence of length $t$ of all strings in $S$. Then a curve $c^*$ that starts and ends at $p_{0}$, has a vertex at $p_{-2}$ for every letter $A$ in $s^*$, a vertex at $p_{2}$ for every letter $B$ in $s^*$, and with vertices at $p_{0}$ between each pair of letters, is within Fr\'echet distance $\delta=1$ from all the curves $\gamma(s_i)$ (refer to Figure~\ref{fig:1dgadgetsD-example}). Note, that $c^*$ has $\ell=2t+1$ vertices. Consider some curve $\gamma(s_i)$, and consider letters of $s^*$ that correspond to letters of $s_i$. We can match the letter gadgets of $\gamma(s_i)$ with the corresponding vertices of $c^*$. Now, consider a letter of $s^*$ that does not have a corresponding letter in $s_i$. Then the corresponding vertex of $c^*$ can be matched to one of the spikes of a buffer gadget that is adjacent to the closest letter gadget of $\gamma(s_i)$. Thus, $d_{DF}(\gamma(s_i),c^*)=1$ for all $i$.
\end{proof}

\begin{lemma}
For any center curve of length at most $2t+1$ within discrete Fr\'echet distance $1$ from all the curves $\gamma(s_i)$, there exists a supersequence of all strings in $S$ of size at most $t$.
\end{lemma}
\begin{proof}
Similarly to the argument above, for every curve $\gamma(s_i)$, the center curve $c$ must have a vertex per each letter gadget, separated by vertices that cover buffer gadgets. Indeed, for a vertex of $c$ to cover an $A$-gadget, it must have a coordinate at most $-2$, and to cover a $B$-gadget, it must have a coordinate at least $2$. Furthermore, two consecutive vertices of $c$ cannot cover two letter gadgets, there must be a vertex in between covering a buffer gadget. Thus, if $c$ has $2t+1$ vertices, then at most $t$ of these vertices can match to letter gadgets of curves $\gamma(s_i)$. A string $s^*$ consisting of letters corresponding to the vertices of $c$ that match to letter gadgets is a supersequence of all strings in $S$. Indeed, for all strings $s_i$, the letter gadgets of $\gamma(s_i)$ corresponding to the letters of $s_i$ are matched to vertices of $c$ in order, and thus they form a subsequence of $s^*$.
\end{proof}
Thus, we can combine the two lemmas into the following theorem.
\begin{theorem}
    The \klc[1] problem in 1D is $\mathsf{NP}$-hard for the discrete Fr\'echet distance.
\end{theorem}

\subsection{Continuous Fr\'echet Distance}\label{sec:1Dhard-cont-app}
Next, we will show how to modify the above $\mathsf{NP}$-hardness construction for the continuous Fr\'echet distance. The subtlety of the continuous case  is that a center curve does not need an extra vertex to cover a buffer gadget when switching between two different letters, i.e., in between $AB$ or $BA$. Thus, for a string of length $t$, depending on its letters, it can be matched to a center curve of size in the range $t+2 \le \ell \le 2t+1$. To make the number of these (possibly) missing vertices comparatively small to the total number of vertices used in covering the curves $\gamma(s_{i})$, and thus to be able to uniquely identify the size of the superstring corresponding to a center curve, we modify the gadgets of our $\mathsf{NP}$-hardness construction in the following way. Let
\[
\begin{aligned}
g_A&:= p_{-3}~ p_{0} \,, & \quad \widehat{g}_A &:= (g_A)^{t-1}~ p_{-3} \\
g_B&:= p_{3}~ p_{0}  \,, & \quad \widehat{g}_B &:= (g_B)^{t-1}~p_{3}\\
g_a&:= p_{-1}~ p_{0} \,, & \quad \widehat{g}_a &:= (g_a)^{t-1}~ p_{-1} \\
g_b&:= p_{1}~ p_{0}  \,, & \quad \widehat{g}_b &:= (g_b)^{t-1}~ p_{1} \,.
\end{aligned}
\]
We map each letter of $s_{i}$ to a one-dimensional curve in the following way:
\[
\begin{aligned}
A &\rightarrow~ (\widehat g_{a}~ \widehat g_{b})^{t}~ \widehat{g}_A~ (\widehat g_{b}~ \widehat g_{a})^{t} \,,\\
B &\rightarrow~ (\widehat g_{b}~ \widehat g_{a})^{t}~ \widehat{g}_B~ (\widehat g_{a}~ \widehat g_{b})^{t} \,.
\end{aligned}
\]
We obtain curve $\gamma(s_{i})$ by concatenating the resulting curves. We call $\widehat g_{A}$ and $\widehat g_{B}$ the letter gadgets, and the subcurves in between---the buffer gadgets. That is, a buffer gadget can consist of $(\widehat g_{a}~ \widehat g_{b})^{t}$, $(\widehat g_{b}~ \widehat g_{a})^{t}$, or a combination of two of these chains.

Similarly to the discrete Fr\'echet distance case, given a curve $c_i$ within continuous Fr\'echet distance $1$ of the curve $\gamma(s_i)$ for some $s_{i}$, the length of the curve $c_i$ must be at least $(2t-1)|s_{i}|+2$ vertices. Indeed, every letter gadget requires $c_i$ to have at least $2t-1$ vertices, as the distance between every two consecutive vertices of a letter gadget is exceeding $\delta=1$. Buffer gadgets that separate letter gadgets may not require additional vertices, but the two buffer gadgets at the ends of $\gamma(s_{i})$ require at least two more vertices.

We can also construct a curve $c'_i$ with $(2t-1)|s_i|+2 \le |c'_i| \le 2t|s_i|+1$ that is within Fr\'echet distance $1$ from $\gamma(s_i)$. Let $c'_i$ start and end at $p_0$, have a vertex at $p_0$ when $\gamma(s_i)$ has a vertex at $p_0$, have a vertex at $p_2$ when $\gamma(s_i)$ has a vertex at $p_3$, and have a vertex at $p_{-2}$ when $\gamma(s_i)$ has a vertex at $p_{-3}$. Furthermore, let $c'_i$ have a vertex at $p_0$ between every pair of the same letter gadgets of $\gamma(s_i)$, i.e., between $AA$ or $BB$. Then, $d_F(c'_i,\gamma(s_i))=1$. Furthermore, one buffer gadget can completely cover $c'_i$, i.e., $d_F(\text{buffer gadget},c'_i)=1$.

Given an instance of the {\sc SCS} problem, we construct curves $\gamma(s_i)$ for all strings $s_i\in S$. Following a similar argument as above, we can prove that there exists a center curve $c$ of size $(2t-1)t+2\le \ell \le 2t^2+1$ within distance one of all curves $\gamma(s_i)$ if and only if there exists a supersequence $s^*$ of $S$ of size $t = \lfloor \frac{\ell + t}{2t} \rfloor$. Thus,
\begin{theorem}
    The \klc[1] problem in 1D is $\mathsf{NP}$-hard for the continuous Fr\'echet distance.
\end{theorem}

\section{Hardness of Approximation in 1D}\label{sec:1Dappx-app}

In this section we will show that the \klc[1] problem is hard to approximate with approximation factor $2-\varepsilon$ for discrete Fr\'echet distance, and $1.5-\varepsilon$ for continuous Fr\'echet distance, for any constant $\varepsilon > 0$.

\subsection{Discrete Fr\'echet Distance}
First, consider the discrete Fr\'echet distance case. The construction we use to prove this is the same as in Section~\ref{sec:1Dhard-discrete-app}. The next lemma is trivial in the discrete Fr\'echet case. Later we will have an analogous lemma for the continuous case as well.
\begin{lemma}\label{lem:no:overlap:1D:discrete-app}
Let $\psi$ and $\phi$ be two one-dimensional point sequences, such that  $d_{DF}(g_A, \psi) < r$ and $d_{DF}(g_B, \phi) < r$ for $r=2$. Each of the following statements is false:
\begin{compactenum}[(i)]
\item $\psi$ contains a prefix that is a suffix of $\phi$,
\item $\phi$ contains a prefix that is a suffix of $\psi$,
\item $\psi$ is a subsequence of $\phi$,
\item $\phi$ is a subsequence of $\psi$.
\end{compactenum}
\end{lemma}

\begin{lemma}\label{lem:1D-discrete-approx-app}
    For any instance of the {\sc SCS} problem, if there exists a center curve of size $\ell$ within discrete Fr\'echet distance $\delta$ from all the curves $\gamma(s_{i})$, where $\delta < 2$, then there exists a center curve of size at most $\ell$ within discrete Fr\'echet distance $1$ from all the curves $\gamma(s_{i})$.
\end{lemma}
\begin{proof}
Denote the center curve within discrete Fr\'echet distance $\delta$ from all the curves $\gamma(s_{i})$ as $c$. By Lemma~\ref{lem:no:overlap:1D:discrete-app}, the subsequences of vertices of $c$ that match to $A$-gadgets and to $B$-gadgets do not overlap. The absolute value of the coordinates of the vertices that match to points $p_{-3}$ and $p_{3}$ of the letter gadgets must be strictly greater than $1$. Thus, these vertices cannot be completely matched to a whole buffer gadget. Two consecutive such points cannot cover a whole buffer gadget either. Thus, between any two points that are covering letter gadgets, there must be at least one vertex covering a buffer gadget. 

Then, consider a curve $c^{*}$ of size at most $\ell$ that starts and ends in $p_{0}$, has a vertex at $p_{-2}$ when $c$ has a vertex with coordinate $<-1$, has a vertex at $p_{2}$ when $c$ has a vertex with coordinate $>1$, and has vertices at $p_{0}$ in between. For any point of a curve $\gamma_{s_{i}}$, if it was within Fr\'echet distance $\delta$ from some vertex of the center curve $c$, then it will be within Fr\'echet distance $1$ from the corresponding vertex $c^{*}$.
\end{proof}
As a corollary we obtain the following theorem.
\hardoneddiscreteappx*

\subsection{Continuous Fr\'echet Distance}
The construction we use to prove this is the same as in Section~\ref{sec:1Dhard-cont-app}.
\begin{lemma}\label{lem:no:overlap:1D:cont-app}
Let $\psi$ and $\phi$ be two one-dimensional curves, such that  $d_F(\widehat{g}_A, \psi) < r$ and $d_F(\widehat{g}_B, \phi) < r$ for $r=1.5$. Each of the following statements is false:
\begin{compactenum}[(i)]
\item $\psi$ contains a prefix that is a suffix of $\phi$,
\item $\phi$ contains a prefix that is a suffix of $\psi$,
\item $\psi$ is a subsequence of $\phi$,
\item $\phi$ is a subsequence of $\psi$.
\end{compactenum}
\end{lemma}
\begin{proof}
We show the proof for $(i)$ and $(iii)$ in detail. The rest can be shown analogously. The coordinate of the first point $q$ in of $\psi$ must be strictly less than $-1.5$, since $\widehat{g}_A$ starts at $p_{-3}$. On the other hand, $\phi$ cannot contain $q$. Indeed, there is no point in $\widehat{g}_B$ that would be within distance $r$ to $q$. Therefore $(i)$ and $(iii)$ are false.
\end{proof}
The following lemma can be proven analogously to Lemma~\ref{lem:1D-discrete-approx-app}.
\begin{lemma}
    For any instance of the {\sc SCS} problem, if there exists a center curve of size at most $\ell$ within discrete Fr\'echet distance $\delta$ from all the curves $\gamma(s_{i})$, where $\delta < 1.5$, then there exists a center curve of size at most $\ell$ within discrete Fr\'echet distance $1$ from all the curves $\gamma(s_{i})$.
\end{lemma}
As a corollary we obtain the following theorem.

\hardonedcontappx*

\section{MEB for continuous Fr\'echet Distance}\label{sec:meb-app}
Finally, we show that the problem of finding a minimum enclosing ball (MEB) is also hard for the continuous Fr\'echet distance. To do this, we use the same construction from Section~\ref{sec:meb}. All that is left is to show is that $(S,t)$ is a true instance of SCS if and only if there exist $(j,j')\in I_t$ such that $(G\cup R_{j,j'}, 1 )$ is a true instance of MEB for the continuous Fr\'echet distance.

\begin{lemma}
If $(S,t)$ is a true instance of SCS, then there exists $(j,j')\in I_t$ such that $(G\cup R_{j,j'}, 1)$ is a true instance of MEB for the continuous Fr\'echet distance.
\end{lemma}
\begin{proof}
Since $(S,t)$ is a true instance of SCS, there exists $(j,j')\in I_t$ such that for the sequence $c^*$ constructed in Lemma~\ref{lem:SCS-implies-j-k-center-curve}, $d_{DF}(c^*,g)\leq 1$ for all $g\in (G\cup R_{j,j'})$. Since the discrete Fr\'echet distance is an upper bound for the continuous Fr\'echet distance, $d_{F}(c^*,g)\leq d_{DF}(c^*,g)\leq 1$ for all $g\in (G\cup R_{j,j'})$.
\end{proof}

\begin{lemma}
If there exists $(j,j')\in I_t$ such that $(G\cup R_{j,j'},1)$ is a true instance of MEB for the continuous Fr\'echet distance, then $(S,t)$ is a true instance of SCS.
\end{lemma}
\begin{proof}
Let $c^*$ be a center curve such that $d_F(g,c^*)\leq 1$ for all $g\in (G\cup R_{j,j'},1)$. For each $p_{-3}$ vertex on $A^j$, there is some point $p$ on the center curve $c^*$ matched to it. Since $p$ has distance at most $1$ to $p_{-3}$ and some point on $B^{j'}$, the point $p$ must be at $p_{-2}$ and the matched point on $B^{j'}$ is a vertex at $p_0$. Since the line segments incident to $p_0$ in $B^{j'}$ both have a distance to $p$ strictly greater than $1$, the neighbourhood around $p$ on $c^*$ cannot be further from $p_0$ than $p$. This  means that either $p$ is a vertex at $p_{-2}$ or lies on a segment connecting two vertices at $p_{-2}$ that are also matched to $p_{-3}$. So, there is an A-part on $c^*$ for every $p_{-3}$ vertex on $A^j$, containing a vertex at $p_{-2}$. Analogously, there is a B-part on $c^*$ for every $p_3$ vertex on $B^{j'}$, containing a vertex at $p_2$. 

Now, analogously to Lemma~\ref{lem:j-k-center-curve-implies-SCS}, we can create a sequence $s^*$ of length $t$ from $c^*$ by replacing the A-parts with the $A$ character and B-parts with $B$ characters. As the letter gadgets in the curves $\gamma(s_i)$ must be matched to some vertex at $p_2$ or $p_{-2}$ representing the A-or B-part, we get that the sequence $s_i\in S$ is a subsequence of $s^*$.
\end{proof}
So, we can combine these two lemmas into the following theorem.

\hardonedMEBcont*